\documentclass[a4paper,UKenglish,cleveref, autoref, thm-restate
]{lipics-v2021}
\pdfoutput=1 
\hideLIPIcs  
\usepackage{mathtools}
    
\usepackage{algorithm}
\usepackage{algpseudocode}
\usepackage{csquotes}
\usepackage{subcaption}

\algnewcommand{\IIf}[1]{\State\algorithmicif\ #1\ \algorithmicthen}
\algnewcommand{\EndIIf}{\unskip\ \algorithmicend\ \algorithmicif}

        \algrenewcommand\algorithmicrequire{\textbf{Input:}}
        \algrenewcommand\algorithmicensure{\textbf{Output:}}

	\theoremstyle{definition}
	\newtheorem{coro}{Corollary}

	\newtheorem{asm}{Assumption}

\usepackage[colorinlistoftodos,textsize=scriptsize,color=blue!25,textwidth=2cm]{todonotes}

\nolinenumbers 
\nolinenumbers 
\nolinenumbers 

\newcommand{\mspan}[2][]{\ensuremath{\textrm{span}_{#1}(#2)}} 

\newcommand{\cI}{\ensuremath{\mathcal{I}}}
\newcommand{\cM}{\ensuremath{\mathcal{M}}}

\title{Optimal Verification of a Minimum-Weight Basis in an Uncertainty Matroid}

\author{Haya Diwan}{Department of Computer Science and Engineering, New York University, USA}{hd2371@nyu.edu}{https://orcid.org/0009-0009-4490-2384}{Partially supported by NSF Grant 1909335.}

\author{Lisa Hellerstein}{Department of Computer Science and Engineering, New York University, USA}{lisa.hellerstein@nyu.edu}{https://orcid.org/0000-0002-3743-7965}{Partially supported by NSF grant 1909335.}

\author{Nicole Megow}{Faculty of Mathematics and Computer Science, University of Bremen, Germany}{nicole.megow@uni-bremen.de}{https://orcid.org/0000-0002-3531-7644}{Supported by DFG grant no.\ 547924951.}
\author{Jens Schlöter}{Centrum Wiskunde \& Informatica (CWI), Amsterdam, The Netherlands}{jens.schloter@cwi.nl}{https://orcid.org/0000-0003-0555-4806}{Supported by the research project \emph{Optimization for and with Machine Learning (OPTIMAL)}, funded by the Dutch Research Council (NWO), grant number OCENW.GROOT.2019.015.}

\authorrunning{H.\ Diwan et al.} 

\Copyright{Haya Diwan, Lisa Hellerstein, Nicole Megow, Jens Schlöter} 



\category{} 

\relatedversion{An extended abstract of this paper will appear in the proceedings of the 43rd International Symposium on Theoretical Aspects of Computer Science (STACS 2026).} 
\begin{document}

\maketitle
 
\begin{abstract}

Research in explorable uncertainty addresses combinatorial optimization problems where there is partial information about the values of numeric input parameters, and exact values of these parameters can be determined by performing costly queries. 
The goal is to design an adaptive query strategy that minimizes the query cost incurred in computing an optimal solution. Solving such problems generally requires that we be able to solve the associated verification problem: given the answers to all queries in advance, find a minimum-cost set of queries that certifies an optimal solution to the combinatorial optimization problem.
We present a polynomial-time algorithm for verifying a minimum-weight basis of a matroid, where each weight lies in a given uncertainty area. These areas may be finite sets, real intervals, or unions of open and closed intervals, strictly generalizing previous work by Erlebach and Hoffman which only handled the special case of open intervals. Our algorithm introduces new techniques to address the resulting challenges.  

Verification problems are of particular importance in the area of explorable uncertainty, as the structural insights and techniques used to solve the verification problem often heavily influence work on the corresponding online problem and its stochastic variant. In our case, we use structural results from the verification problem to give a best-possible algorithm for a promise variant of the  corresponding adaptive online problem. Finally, we show that our algorithms can be applied to two learning-augmented variants of the minimum-weight basis problem under explorable uncertainty.

\end{abstract}

\newpage
\setcounter{page}{1}

\section{Introduction}

Consider the problem of finding a minimum-weight basis (MWB) of a matroid, under uncertainty. We consider a weighted matroid $M = (E, \mathcal{I}, w)$, where each element $e \in E$ has a weight $w_e$ that is known to be contained within a given range, called an \emph{uncertainty area} $A_e$. The exact value of $w_e$ can be obtained by performing a query on element $e$ that incurs a non-negative cost of $c_e$. The problem is to determine an MWB while minimizing the total query cost.  In the special case where $M$ is the graphic matroid of a connected graph, so the bases of $M$ correspond to spanning trees, this is the minimum spanning tree (MST) problem under uncertain edge weights.  
This was one of the earliest studied problems in explorable uncertainty, motivated by network design: spanning trees capture fundamental connectivity tasks, and uncertainty areas naturally model edge weights known only within ranges, such as costs, latencies, or capacities.

The verification version of the MWB problem under uncertainty is an offline problem that, given access to both the uncertainty areas $A_e$ and the exact weights $w_e$, seeks a minimum-cost set of queries for verifying some minimum-weight basis of $M$. More particularly, these queries and their answers, together with the uncertainty areas for all elements $e \in E$ not queried, would be sufficient to determine an MWB 
of $M$.

To illustrate the MWB verification problem, consider the special case of a graphic matroid of a connected graph $G$. The elements of this matroid are the edges in $G$, the independent sets are the cycle-free edge sets, and an MWB is an MST.  
Figure~\ref{fig: open vs closed} shows three uncertainty graphs with different types of intervals as their uncertainty areas, and their corresponding minimum-cost query sets, assuming unit query cost per edge.

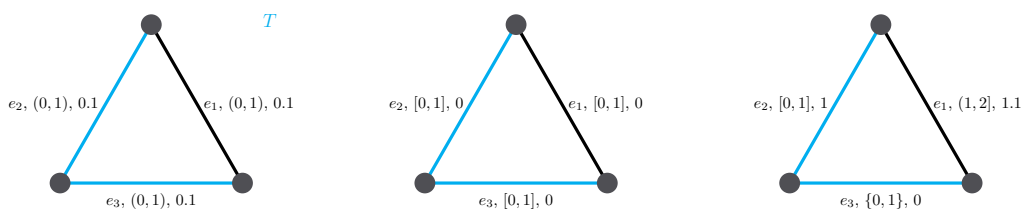
\begin{figure}[h!]
    \centering
    \scalebox{0.6}{
        \begin{tikzpicture}
            \tikzstyle{vertex} = [circle, fill=darkgray, minimum size=13pt, inner sep=0pt]
            \tikzstyle{edge} = [line width=2pt] 

            \begin{scope}[shift={(-8, 0)}]
                \coordinate (A) at (0, 0);
                \coordinate (B) at (2, 3.5);
                \coordinate (C) at (4, 0);

                \draw[edge, cyan] (A) -- (B) node[midway, left, black] {$e_2$, \((0,1)\), 0.1};
                \draw[edge, cyan] (A) -- (C) node[midway, below, black, yshift=-2pt] {$e_3$, \((0,1)\), 0.1};
                \draw[edge] (C) -- (B) node[midway, right, black] {$e_1$, \((0,1)\), 0.1};

                \foreach \p in {A,B,C} {
                    \node[vertex] at (\p) {};
                }
            \end{scope}

            \begin{scope}[shift={(0, 0)}]
                \coordinate (A) at (0, 0);
                \coordinate (B) at (2, 3.5);
                \coordinate (C) at (4, 0);

                \draw[edge, cyan] (A) -- (B) node[midway, left, black] {$e_2$, \([0,1]\), 0};
                \draw[edge, cyan] (A) -- (C) node[midway, below, black, yshift=-2pt] {$e_3$, \([0,1]\), 0};
                \draw[edge] (C) -- (B) node[midway, right, black] {$e_1$, \([0,1]\), 0};

                \foreach \p in {A,B,C} {
                    \node[vertex] at (\p) {};
                }
            \end{scope}

            \begin{scope}[shift={(8, 0)}]
                \coordinate (A) at (0, 0);
                \coordinate (B) at (2, 3.5);
                \coordinate (C) at (4, 0);

                \draw[edge, cyan] (A) -- (B) node[midway, left, black] {$e_2$, \([0, 1]\), 1};
                \draw[edge, cyan] (A) -- (C) node[midway, below, black, yshift=-2pt] 
                {$e_3$, $\{0,1\}$, 0};
                \draw[edge] (C) -- (B) node[midway, right, black] {$e_1$, \((1,2]\), 1.1};

                \foreach \p in {A,B,C} {
                    \node[vertex] at (\p) {};
                }
            \end{scope}

            \node[cyan, font=\Large\bfseries] at (-3.4, 3.6) {$T$};

        \end{tikzpicture}
    }
    \caption{
    Let $T$ (blue) be an MST of the given graph. The figure shows the minimum-cost query sets for instances with different types of uncertainty intervals: open \((0, 1)\), closed \([0, 1]\), or mixed. Edge labels indicate (name, interval, weight), and each query has unit cost. The corresponding minimum-cost query sets are $Q = \{e_1, e_2, e_3\}$ (open), $Q = \{e_2, e_3\}$ (closed), and $Q = \emptyset$ (mixed).}
    \label{fig: open vs closed}
\end{figure}

Our main result is a polynomial-time algorithm 
for this verification problem, producing both a minimum-cost verification set (we also say {\em certificate}), and an associated MWB~$B$, for uncertainty areas that are a finite union of intervals.  Each of these intervals can be open or closed.  Note that this includes the case where each uncertainty area is a finite set of discrete values. For example, 
$\{0,1\}$ is equivalent to $[0,0] \cup [1,1]$.



In addition to being an interesting combinatorial problem in its own right, the verification problem is fundamental to solving variants of the online MWB problem in settings beyond the worst-case, such as stochastic settings and adversarial learning-augmented variants. 
In the stochastic setting, each weight $w_e$ is assumed to be drawn from a known probability distribution, and the goal is to minimize the \emph{expected} query cost.    
The verification problem can be viewed as either an analog or a special case of such a stochastic formulation, where each distribution has support of size~1, but an element $e$ still must be queried in order to ``use'' the value $w_e$. 
At the same time, our results have direct consequences for variants of the adversarial adaptive online problem: in Section~\ref{sec: applications} we  use insights from the verification problem to obtain new results for a variant of the online adaptive MWB problem, and show how it provides a foundation for learning-augmented algorithms. Thus, verification serves both as a crucial first step toward stochastic formulations and as a direct tool for online problems under uncertainty.

A special case (in two senses) of our verification problem was studied before by Erlebach and Hoffman~\cite{erlebach_verification}. They considered finding the MST of a graph under weight uncertainty, with uncertainty areas that are either open intervals, i.e., $A_e = (L_e,U_e)$, or trivial, i.e., $A_e=\{w_e\}$ (meaning the weight of $e$ is given). They gave an efficient algorithm that computes an optimal solution.  

From a technical standpoint, handling uncertainty areas that are closed intervals and finite sets, as we do in our work, posed new challenges and required new techniques as well as structural insights.  We discuss 
differences and challenges 
in more detail~below.

\subsection*{Related Work}

\noindent\textbf{Verification problems\ } In the context of verifying optimal structures under uncertainty, the above-mentioned polynomial-time algorithm of Erlebach and Hoffman~\cite{erlebach_verification}, for MST verification with open and trivial uncertainty areas, stands out as a notable exception in terms of computational complexity.
In contrast, other verification problems are known to be NP-hard and even inapproximable. 
For instance, verification is NP-hard for identifying the set of maximal points under geometric uncertainty~\cite{CharalambousH13}, and for selecting the cheapest set from a collection of sets with elements of uncertain weight~\cite{Erlebach0K16}. 
The latter is even NP-hard to approximate within a factor of $o(\log m)$, where $m$ is the number of sets~\cite{MegowS23}. Recently, verification for knapsack under explorable uncertainty was shown to be $\Sigma_2^p$-complete~\cite{schloter25}.


\medskip 
\noindent\textbf{Online algorithms\ } The line of research on exploring uncertain numerical values under query costs was initiated by Kahan~\cite{kahan}, who studied problems such as identifying the minimum or $k$-th smallest value from a set of uncertain values. Since then, many problems have been investigated~\cite{erlebach15querysurvey}, mainly in an \emph{adaptive} online setting, where queries can be performed sequentially based on prior outcomes. The goal is to minimize the worst-case competitive ratio, comparing the total cost of an online algorithm to the minimum-cost of a verification set. Favorable constant (and often matching) upper and lower bounds are known for selection-type problems~\cite{kahan,feder_median}, MST and matroid bases~\cite{erlebach08steiner_uncertainty,MegowMS17,MerinoS19} and  sorting~\cite{HalldorssonL21sorting}. However, strong lower bounds have been shown for more complex problems such as cheapest set \cite{Erlebach0K16}, knapsack, matchings, linear programming~\cite{meissner18querythesis}, which rule out any constant competitive ratio.

The online adaptive MWB problem under uncertainty was studied by Erlebach et al.~\cite{Erlebach0K16} for open and trivial intervals. They gave a deterministic algorithm with a best-possible competitive ratio of 2.  This result extended earlier work~\cite{erlebach08steiner_uncertainty} on finding the MST in a graph with 
the same types of uncertainty areas. Subsequently, Megow et al.\ gave  a randomized algorithm with an improved 
bound of $1.707$ for the online versions of both MST and MWB problems, 
again with open and trivial intervals~\cite{MegowMS17}.
In~\cite{MathwieserC24}, Mathwieser and {\c{C}}ela gave improved randomized algorithms for special cases of the problem.
Note that it is crucial to omit closed intervals in this online setting as otherwise no constant competitive ratio is possible~\cite{erlebach08steiner_uncertainty}.

In contrast to the above, \emph{non-adaptive} online variants have received little attention, with the notable exception of the work of Merino and Soto~\cite{MerinoS19}. They gave an algorithm that
computes a  minimum-cost  \enquote{universal} set of queries that can be used to find an MWB of an uncertainty matroid, assuming no prior knowledge of any of the actual element weights. 
Their algorithm, like ours, allows for uncertainty areas that are finite unions of bounded real intervals that can be either open or closed.
However, they require a query set to be a certificate for \emph{every} weight vector $w'$ such that $w'_e \in A_e$ for all elements $e$, while  the verification problem requires that it be a certificate \emph{only} for the given weights $w_e$. Because of the stronger requirements, the minimum-cost certificate for their problem can have significantly higher cost.
In fact, using their algorithm for the verification problem (ignoring access to the weights~$w_e$) only yields an $n$-approximation for uniform query costs and an unbounded approximation factor for general query costs (cf.~\Cref{app:further:discussion}). Thus using the given weights $w_e$ is crucial to optimally solving the verification problem.

\medskip 
\noindent\textbf{Beyond-worst-case models and the relevance of verification methods\ } The strong lower bounds for more complex problems suggest that the adversarial online model may be overly pessimistic. 
This has motivated the study of beyond-worst-case approaches in the context of explorable uncertainty, such as stochastic models~\cite{BampisDEdLMS21,ChaplickHLT20,MegowS23}, where query outcomes are drawn from (possibly unknown) probability distributions, and learning-augmented frameworks~\cite{ErlebachLMS22,ErlebachLMS23}, where algorithms have access to imperfect predictions of the actual weights. 
In all of these works, a solid understanding of the corresponding verification problem was fundamental to the design and analysis of algorithms.

\subsection*{Our Results}

We give a polynomial-time algorithm that solves the MWB verification problem, producing both a minimum-cost verification set (a certificate) and an associated MWB, for uncertainty areas that are a finite union of bounded intervals. Each interval can be either open or closed.

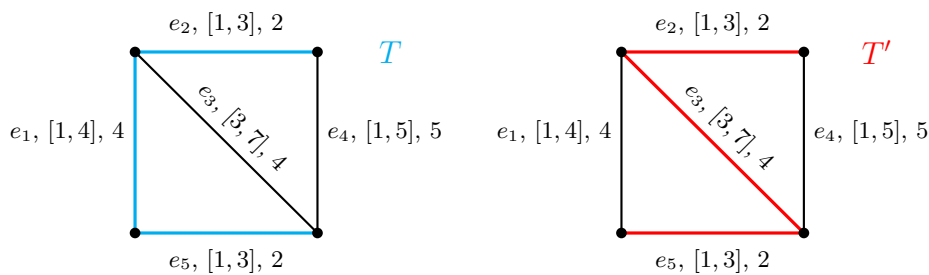
\begin{figure}[h!]
    \centering
    \scalebox{1}{
        \begin{tikzpicture}[scale=0.8, every node/.style={font=\small}]
            \tikzset{vertex/.style={circle, fill=black, minimum size=4pt, inner sep=0pt}}
            \tikzset{edge/.style={thick}}

            \begin{scope}[shift={(0, 0)}]
                \coordinate (A) at (0,0);
                \coordinate (B) at (0,3);
                \coordinate (C) at (3,3);
                \coordinate (D) at (3,0);

                \draw[edge, cyan, very thick] (A) -- (B) node[midway, left, black, yshift=4pt] {$e_1$, \([1,4]\), 4};
                \draw[edge, cyan, very thick] (B) -- (C) node[midway, above, black, yshift=2pt] {$e_2$, \([1,3]\), 2};
                \draw[edge, cyan, very thick] (A) -- (D) node[midway, below, black, yshift=-2pt] {$e_5$, \([1,3]\), 2};
                \draw[edge] (C) -- (D) node[midway, right, black, yshift=4pt] {$e_4$, \([1,5]\), 5};
                \draw[edge] (B) -- (D) node[midway, sloped, above, black] {$e_3$, \([3,7]\), 4};

                \foreach \p in {A,B,C,D} {
                    \node[vertex] at (\p) {};
                }

                \node[cyan, font=\Large\bfseries] at (4.2, 3) {$T$};
            \end{scope}

            \begin{scope}[shift={(8, 0)}] 
                \coordinate (A) at (0,0);
                \coordinate (B) at (0,3);
                \coordinate (C) at (3,3);
                \coordinate (D) at (3,0);

                \draw[edge] (A) -- (B) node[midway, left, black, yshift=4pt] {$e_1$, \([1,4]\), 4};
                \draw[edge, red, very thick] (B) -- (C) node[midway, above, black, yshift=2pt] {$e_2$, \([1,3]\), 2};
                \draw[edge, red, very thick] (A) -- (D) node[midway, below, black, yshift=-2pt] {$e_5$, \([1,3]\), 2};
                \draw[edge] (C) -- (D) node[midway, right, black, yshift=4pt] {$e_4$, \([1,5]\), 5};
                \draw[edge, red, very thick] (B) -- (D) node[midway, sloped, above, black] {$e_3$, \([3,7]\), 4};

                \foreach \p in {A,B,C,D} {
                    \node[vertex] at (\p) {};
                }

                \node[red, font=\Large\bfseries] at (4.2, 3) {$T'$};
            \end{scope}
        \end{tikzpicture}
    }
    \caption{Let $T$ (blue) and $T’$ (red) be MSTs of the same uncertainty graph. Edge labels show (name, interval, weight), and all queries have unit cost. While both are MSTs, $T$ has a cheaper minimum-cost query set, $Q = \{e_3, e_4\}$, compared to $T’$ with $Q = \{e_1, e_3, e_4\}$.}
    \label{fig: diff queries for diff MST}
\end{figure}

Our algorithm has two phases, motivated by a key difference from the verification problem for open intervals: If uncertainty areas contain~closed intervals, the cost of verification may vary across MWBs; 
see \Cref{fig: diff queries for diff MST}. 
In contrast, 
for open intervals, all MWBs have the same minimum verification cost\footnote{This fact is only implicit in previous works~\cite{erlebach_verification,ErlebachLMS22}. 
However, it is also a corollary of our results in~\Cref{sec: identify MWB}.}.
This raises a natural question for general uncertainty areas: Can we efficiently identify an MWB whose minimum verification cost is as small as possible? (We make the standard assumption that the matroid is given by an independence oracle, and that oracle queries are answered in time polynomial in the number of matroid elements.)

Our first main result answers this question in the affirmative.  The first phase of our algorithm is a polynomial-time procedure that identifies an MWB which has minimum possible verification cost (\Cref{sec: identify MWB}).  Our algorithm is based on carefully designed contraction and deletion rules that allow us to delete and contract \emph{extreme case} elements, i.e., elements $e$ with $w_e = \inf A_e$ or $w_e = \sup A_e$, while maintaining the invariant that there exists an MWB consistent with the deletions and contractions having  minimum possible verification cost. 
As a corollary, we show that,
for uniform query costs and uniform areas $A_e =\{L,U\}$ with $L < U$ (i.e., each $w_e$ is either equal to $L$ or $U$), all MWBs have the same minimum verification cost.

In the second phase (\Cref{sec: construct cert}), our algorithm efficiently computes a minimum-cost certificate to verify a given MWB.\
Together, the two phases of our algorithm yield a polynomial-time algorithm
for our MWB verification problem.

Motivating the second phase, we give a full structural characterization of certificates for verifying an MWB.\ In contrast to previous work, we have to carefully handle  extreme case elements. Intuitively, the presence or absence of such elements can make a huge difference for the certificates. For example, if we want to find an element of maximum weight in a set of elements with identical uncertainty areas, then querying a single element $e$ with $w_e = \sup A_e$ is sufficient to identify an element of maximum weight, whereas the absence of extreme case elements forces us to query all elements. Based on our certificate characterization, we compute a minimum-cost certificate by solving a minimum-weight vertex cover problem in a bipartite auxiliary graph; this technique  has also been used in previous work~\cite{erlebach_verification,MegowMS17,ErlebachLMS22,ErlebachLMS23,BampisDEdLMS21}.  
We remark that our results imply an alternative, arguably simpler, algorithm for the verification problem with open intervals (and other special cases where all MWBs have the same verification cost): Fix any MWB and 
run the algorithm's second phase.

Finally, we use insights from the second phase of our verification algorithm to give new results for the online adaptive MWB problem with general uncertainty areas. For open uncertainty intervals, the competitive ratio of the online adaptive MWB problem is~$2$~\cite{erlebach08steiner_uncertainty}. If closed uncertainty areas are allowed, the competitive ratio increases to $n$~\cite{GuptaSS16}. We show that the competitive ratio of $2$ can be recovered in a promise variant of the online MWB problem with general uncertainty areas, where the algorithm is given an MWB and only has to verify this MWB. This shows that the increase in the competitive ratio from {\em open} to {\em general} uncertainty areas stems from the algorithm's task of \emph{finding} rather than merely \emph{verifying} an MWB $B$. Based on this insight, we design a learning-augmented algorithm for a setting where algorithms have access to an untrusted prediction of an MWB.

\section{Preliminaries} \label{sec: prelim}

We assume familiarity with matroids and only briefly define basic concepts and notation. For a comprehensive  introduction, 
see Schrijver's book~\cite{Schrijver2003book}. Readers might find it helpful 
to keep the graphic matroid and the MST problem on a connected graph in mind. 
For a set~$X$ and an element $e$, we use the short notation $X+e \coloneqq X\cup\{e\}$ and $X-e \coloneqq X\setminus{\{e\}}$.

\subsection{Matroid basics}

A {\em matroid} is a non-empty, downward-closed set system $(E, \cI)$ with element set~$E$ and a family of subsets $\cI\subseteq 2^{E}$, which satisfies the {\em augmentation property}: if $I,J \in \cI $ and $|I|<|J|$, then $I+e\in  \mathcal{I}$ for some $e\in J\setminus{I}$.

Given a matroid $M=(E,\cI)$, a set $I\subseteq E$
is called \emph{independent} if $I\in \cI$, and \emph{dependent} otherwise. We refer to the elements of $M$ by $E(M)$.
An inclusion-wise maximal independent subset is called a \emph{basis} of~$M$. With a matroid $M$, we associate a rank function $r: 2^{E} \rightarrow \mathbb Z_{\ge 0}$, where $r(X)$ describes the maximal cardinality of an independent subset of $X$.

The {\em span} of some $X \subseteq E$ contains all elements that  do not increase the rank when added to $X$, i.e., $\mspan[M]{X} = \{e \in E \mid r(X) = r(X + e)\}$. If $X$ is a basis and $e' \not\in \mspan[M]{X-e}$ for some $e \in X$, then the definition of the span implies that $X-e+e'$ is a basis.

A {\em circuit} $C$ in a matroid $M=(E,\cI)$ is a minimally dependent set, that is, $C \not\in I$ whereas $C\setminus \{e\} \in I$ for each $e \in C$. Any independent set $I\in \cI$ of $M$ and an $e \in E$ such that  $I +e$ is dependent, form a unique circuit  $C_e^I \subseteq I + e$. This circuit is called the {\em fundamental circuit} of $e$ with respect to  $I$. 
If $B$ is a basis and $e \in C_{e'}^B$ for some $e' \not\in B$, then the definition of fundamental circuits implies that  $B' = B-e+e'$ is also a basis of $M$.

Given a matroid $M=(E, \cI)$, we obtain a matroid $M'=(E\setminus X, \cI')$ from $M$ by applying the operations `contraction' and `deletion' for subsets $X\subseteq E$. The {\em deletion} of $X$ from $M$ yields the matroid $M'$ with $\cI':= \{I \,|\, I \subseteq E\setminus X, I \in \cI \}$. In this paper, we only consider deletions that do not decrease the rank of the matroid, i.e., satisfy $r(M) = r(M')$.
{\em Contracting} an independent set $X$ in $M$ yields a matroid $M'$ with independent sets $\cI'$ defined as subsets $I\subseteq E\setminus X$ with $I\cup X \in \cI$. A matroid $M'$ obtained from $M$ by a series of contractions and deletions is called a {\em minor of $M$}.

We consider {\em weighted} matroids $M = (E, \mathcal{I}, w)$, in which each element $e\in E$ has an associated weight $w_e\in \mathbb{R}$. 
The weight of a basis $B$ of $M$ is $w(B)\coloneqq \sum_{e\in B}w_e$.  We seek a minimum-weight basis (MWB) of $M$. Note that an MWB of $M$ may not~be~unique.

\subsection{Properties of bases}

Throughout the paper, we use 
well-known properties to argue that a basis $B$ is an MWB. The 
propositions follow, e.g., from~\cite[Thm.~6.1]{lawler2001combinatorial}, and the observation is proved in~\cref{app:prelim}.

\begin{proposition}
    \label{prop:min:cut}
    Let $B$ be a basis of the matroid $M=(E,\cI,w)$. If $w_e \le w_{e'}$ holds for all $e \in B$ and all $e' \in E\setminus \mspan[M]{B-e}$, then $B$ is an MWB of $M$.
\end{proposition}

\begin{proposition}
    \label{prop:max:circuit}
    Let $B$ be a basis of the matroid $M=(E,\cI,w)$. If $w_e \ge w_{e'}$ holds for all $e \in E\setminus B$ and all $e' \in C^B_e$, then $B$ is an MWB of $M$.
\end{proposition}
\begin{restatable}{obs}{ObsMatroidCuts}
    \label{obs:matroid:cuts}
    Let $B$ be a basis of the matroid $M = (E, \cI,w)$. If $C$ is a circuit of $M$ with $e \in B \cap C$, then there is an $e' \in C-e$ such that $e' \notin \mspan[M]{B-e}$. 
\end{restatable}

\subsection{Uncertainty matroids and certificates}

In an uncertainty matroid, weights $w_e$ are not known, but we are given uncertainty areas guaranteed to contain them.

\begin{definition}[Uncertainty Matroid] An uncertainty matroid $\cM=(E,\cI,A)$ is a matroid $M=(E,\cI)$ and a function $A: E \rightarrow \mathbb{R}$ such that for each $e \in E$, $A(e)$ is a non-empty finite union of bounded real intervals (either open or closed).   
We call $\cM=(E,\cI,A,w)$ a {\em weighted} uncertainty matroid if $\cM=(E,\cI,A)$ is an uncertainty matroid and
$w:E \rightarrow \mathbb{R}$ is such that $w_e := w(e) \in A_e$ for all $e \in E$. We define $A_e:=A(e)$, $L_e := \inf A_e$, and $U(e) :=\sup A_e$.
We call $A_e$ the uncertainty area of $e$ and say $e$ is {\em trivial} if $A_e$ consists of the single value $w_e$.  
\end{definition}

Consider the problem of computing an MWB in a weighted uncertainty matroid, where the exact weights are  initially unknown 
but the true weight of an element $e$ can be revealed through a {\em query}, which incurs a cost $c_e \geq 0$.
We are now faced with the problem of determining an MWB at the lowest possible total query cost. We use the notion of \emph{certificates}, where a certificate is a set $Q \subseteq E$ such that querying $Q$ reveals sufficient information about the uncertain weights in $A$ to identify a basis $B$ that is an MWB, irrespective of the exact weights of the elements not in $Q$. We make this precise in  the definitions below.  
Note that the formal definition of an uncertainty matroid only specifies uncertainty areas for the matroid elements $e$, not specific weights $w_e$. 
In what follows, we sometimes identify a weighted uncertainty matroid $\cM(E,\cI,A,w)$ with its associated weighted matroid $M=(E,I,w)$.  For example, when we refer to a basis $B$ of $\cM(E,\cI,A,w)$, we mean a basis of $M=(E,I,w)$.

\begin{definition}[Consistent weight assignment]
    \label{def: weight assignment consistent with certificate}
    Consider a weighted uncertainty matroid $\mathcal{M} = (E, \mathcal{I}, A, w)$.
    A \emph{weight assignment} $w^*$ for $\mathcal{M}$ is a function $w^*\colon E \rightarrow \mathbb{R}$ such that $w^*_e \in A_e$ holds for each $e \in E$.  We say a weight assignment $w^*$ for $\mathcal{M}$ is \emph{consistent} with $Q$ if it additionally satisfies  $w_e = w^*_e$ for each $e \in Q$.
\end{definition}

\begin{definition}[Verification and certificates]
    \label{def: what it means for S to verify T}
    Consider a weighted uncertainty matroid $\mathcal{M} = (E, \mathcal{I}, A,w)$.
    Let $B$ be an MWB of $\cM$. We say that a certificate $Q$ {\em verifies} $B$ if for every weight assignment $w^*$ for $\cM$ that is consistent with $Q$ (Definition \ref{def: weight assignment consistent with certificate}), $B$ is an MWB of the weighted matroid $M'=(E,\cI,w^*)$. 
\end{definition}

\begin{definition}[Minimum-cost certificate] 
Consider a weighted uncertainty matroid $\cM = (E, \mathcal{I}, A,w)$. 
Let $\mathcal{Q}$ be the set of all certificates that verify a basis of $\cM$. Then, $Q\in \mathcal{Q}$ is a {\em minimum-cost certificate} for $\cM$, if 
$c(Q) \coloneqq \sum_{e \in Q} c_e$ is minimal among all certificates in $\mathcal{Q}$.
\end{definition}

\subsubsection{Characterization of certificates}
To further characterize certificates, define $L_e(Q) = w_e$ if $e \in Q$ and $L_e(Q) = L_e$ if $e \not\in Q$. Similarly, define $U_e(Q) = w_e$ if $e \in Q$ and $U_e(Q) = U_e$ if $e \not\in Q$.
Intuitively, $L_e(Q)$ and $U_e(Q)$ denote the upper and lower limits of $e$ \emph{after} querying the certificate $Q$.

\begin{lemma}
    \label{lem:certificate:characterization:cuts}
    Let $B$ be an MWB of 
    $\mathcal{M}= (E,\cI, A,w)$. A set $Q \subseteq E$ is a certificate that verifies $B$  if and only if $U_e(Q) \le L_f(Q)$ for all $e \in B$ and $f \in (E\setminus \mspan[M]{B-e})-e$.
\end{lemma}

\begin{proof}
    First, assume $U_e(Q) \le L_f(Q)$ for all $e \in B$ and $f \in (E\setminus \mspan[M]{B-e})-e$. Consider any weight assignment $w^*$ consistent with $Q$.
    Then, for each $e \in B$, we have $w^*_e \le U_e(Q) \le L_f(Q) \le w^*_f$ for all $f \in (E\setminus \mspan[M]{B-e})-e$. This implies $B$ is an MWB for $w^*$, by~\Cref{prop:min:cut}. 
    Thus $Q$ verifies $B$ by~\Cref{def: what it means for S to verify T}.

    Next, assume it does not hold that  $U_e(Q) \le L_f(Q)$ for all $e \in B$ and $f \in (E\setminus \mspan[M]{B-e})-e$. Let $e \in B$ and $f \in (E\setminus \mspan[M]{B-e})-e$ be such that $L_f(Q) < U_e(Q)$. Then, there exists a weight assignment $w^*$, consistent with $Q$, such that $w^*_e > w^*_f$. Thus $B' = B - e + f$ is independent and satisfies $w^*(B') < w^*(B)$, and so $Q$ does not verify $B$  by~\Cref{def: what it means for S to verify T}.
\end{proof}

\noindent The lemma implies the following alternative characterization, which we 
prove in~\Cref{app:prelim}.

\begin{restatable}{coro}{coroCharacterization}
            \label{coro:certificate:characterization:circuits}
        Let $B$ be an MWB of $\mathcal{M}= (E,\cI, A,w)$.
        A set $Q \subseteq E$ is a certificate that verifies~$B$ if and only if the following property holds for every $f \not \in B$, and the fundamental circuit $C_f^B$: 
        for every $e \in C_f - f$, $U_e(Q) \leq L_f(Q)$.
\end{restatable}

\subsubsection{Exchange properties of certificates} We continue by showing some exchange properties of certificates. 
\label{sec:exchange:properties}

\begin{lemma}
    \label{lem:extreme:upperlimit:exchange}
    Let $B$ be an MWB of $\mathcal{M}= (E,\cI, A,w)$
    and let $Q$ be a certificate that verifies $B$. Let $e \in B$ and $e' \not\in B$ with $U_e = U_{e'} = w_e = w_{e'}$ be such that $e \in C_{e'}$ for the fundamental circuit $C_{e'}$ of $e'$ with respect to $B$. Then $Q' = Q - e' + e$ is a certificate for $B' = B - e + e'$.  
\end{lemma}

\begin{proof}
    During this proof, we use $C'_f$ to refer to the fundamental circuit of $f$ with respect to $B'$ for an $f \not\in B'$.
    To show that $Q'$ is a certificate for $B'$, we show for each $f \not\in B'$ that $L_f(Q') \ge U_{f'}(Q')$ for each $f' \in C'_f - f$. 
    By~\Cref{coro:certificate:characterization:circuits}, this implies $Q'$ is a certificate for $B'$. 
    
    We distinguish between the two cases (1) $f=e$ and (2) $f \not= e$.

    \textbf{Case (1):} Assume $f = e$ and consider circuit $C'_e$. Since $B' = B - e + e'$, we have $C'_e = C_{e'}$.  
    Since $Q$ is a certificate for $B$ with $e' \not\in B$, we have $w_{e'} \ge L_{e'}(Q) \ge U_{f'}(Q)$ for all $f' \in C_{e'} - e'$ by~\Cref{coro:certificate:characterization:circuits}. For every $f' \in C_{e'} - \{e,e'\}$, the membership of $f$ and $f'$ in $Q$ and $Q'$ is identical,  so $U_{f'}(Q') = U_{f'}(Q)$.
    Since $C_{e'} = C'_e$, $e \in Q'$ and by definition $w_e = w_{e'}=U_{e'}$, this implies $L_e(Q')=w_e=w_{e'} \ge L_{e'}(Q) \ge U_{f'}(Q) = U_{f'}(Q')$ for all $f' \in C'_{e} - \{e,e'\}$.  Further, $L_e(Q')=w_{e'}=U_{e'} \geq U_{e'}(Q')$.
    
    \textbf{Case (2):} Assume $f \not= e$ and consider the circuit $C'_f$. We distinguish two subcases. 

    \begin{itemize}
        \item If $(C'_f -f) \subseteq B$, then $C'_f = C_f$ and $e,e' \not\in C'_f$. In this case, $L_f(Q') = L_f(Q) \ge U_{f'}(Q) = U_{f'}(Q')$ holds for all $f' \in C'_f-f$ by definition of $Q'$ and by our assumption that $Q$ is a certificate for $B$ with $f \not\in B$ (cf.~\Cref{coro:certificate:characterization:circuits}).
        
        \item If  $(C'_f -f) \not\subseteq B$, then we must have $e' \in C'_f - f$.  Consequently, $f \in E\setminus \mspan[M]{B'-e'} = E\setminus \mspan[M]{B-e}$.  Note that $w_{e'} = U_{e'}(Q') = U_{e}(Q)$ holds by definition because of $w_e=w_{e'}=U_e=U_{e'}$. Since $Q$ is a certificate for $B$ with $e\in B$, we must have $w_{e'} = U_{e'}(Q') = U_{e}(Q) \le L_f(Q) = L_f(Q')$ by~\Cref{lem:certificate:characterization:cuts}. 

        If $C'_f = \{f,e'\}$, then the above argument already shows $L_f(Q') \ge U_{f'}(Q')$ for each $f' \in C'_f - f$, and we are done with the proof for this case. Thus, assume $C'_f \setminus \{f,e'\} \not= \emptyset$ and consider an arbitrary $f' \in C'_f - f - e'$.
        Then $f' \in B \cap C'_f$ and, by~\Cref{obs:matroid:cuts}, there must be an element $g \in C'_f - f'$ with $g \not\in \mspan[M]{B-f'}$. The only two elements of $C_f' - f'$ that are not in $B-f' \subseteq \mspan[M]{B-f'}$ are $e'$ and $f$.

        If $f \not\in \mspan[M]{B-f'}$, then our assumption that $Q$ is a certificate for $B$ implies $U_{f'}(Q') = U_{f'}(Q) \le L_f(Q) = L_f(Q')$.
        If $e' \not\in \mspan[M]{B-f'}$, then our assumption that $Q$ is a certificate for $B$ implies $U_{f'}(Q') = U_{f'}(Q) \le L_{e'}(Q) \le w_{e'} = U_{e}(Q) \le L_f(Q')$, where the last inequality uses 
        $U_{e}(Q) \le L_f(Q')$ as argued above. In all cases, 
        $U_{f'}(Q') \le L_f(Q')$. \qedhere
    \end{itemize}
\end{proof}

The next lemma is a dual version of~\Cref{lem:extreme:upperlimit:exchange}.  The proof is in~\Cref{app:prelim}.

\begin{restatable}{lemma}{lemLowerLimitExchange}
    \label{lem:extreme:lowerlimit:exchange}
    Let $B$ be an MWB of $\mathcal{M}= (E,\cI, A,w)$
    and let $Q$ be a certificate that verifies $B$. Let $e \notin B$ and $e' \in B$ with $L_e = L_{e'}=w_e = w_{e'}$ be such that $e' \in C_{e}$ for the fundamental circuit $C_{e}$ of $e$ with respect to $B$. Then $Q' = Q - e' + e$ is a certificate for $B' = B + e - e'$.
\end{restatable}

It is easy to show the following exchange property for trivial elements.

\begin{restatable}{lemma}{lemTrivialExchange}
     \label{lem:trivial:exchange}
    Let $B$ be an MWB of $\mathcal{M}= (E,\cI, A,w)$
    and let $Q$ be a certificate that verifies~$B$.
    Let $e \in B$ and $e' \not\in B$ with $w_e = w_{e'}$ be such that $B' = B-e+e'$ is independent.
    If each $f \in \{e,e'\}$ is either trivial or contained in $Q$, then $Q$ also verifies $B'$.
\end{restatable}
\section{Computing an MWB that has a minimum-cost certificate} \label{sec: identify MWB}

Given a weighted uncertainty matroid $\mathcal{M} = (E, \mathcal{I}, A,w)$, we compute an MWB~$B$ that admits a minimum-cost certificate, though we do not yet compute the  certificate itself. To this end, we introduce a set of contraction and deletion rules. Starting from $B= \emptyset$, these rules allow us to safely add elements to $B$ (contraction) and ban elements from ever entering $B$ (deletion) while maintaining the invariant that there is a minimum-cost certificate which verifies a basis $B'$ that contains $B$ 
and does not contain any deleted elements. We formalize this invariant by introducing the notion of a ``compatible minor" which intuitively is the remaining (uncertainty) matroid after applying contractions and deletions.
We refer to sets of contracted and deleted elements as $K$ and $D$, respectively.

\begin{definition}
    Let $\cM=(E,\cI,A,w)$ be a weighted uncertainty matroid with a minimum-cost certificate of cost $c^*$ and let $D,K \subseteq E$. Let $M[D,K]$ be the matroid obtained from $M=(E,\cI)$ by deleting $D$ and contracting $K$. 
    We say $M[D,K]$ is a {\em compatible minor} of  $\cM$ if there is a query set for $\cM$ with cost $c^*$ that verifies an MWB $B$ of 
    $\cM$ with $K \subseteq B$ and $D \cap B = \emptyset$.
    Furthermore, we say that a certificate $Q$ for $\cM$ is \emph{compatible with} $M[D,K]$ if~$Q$ verifies a basis $B$ of 
    $\cM$ with $K \subseteq B$ and $D \cap B = \emptyset$.
\end{definition}

Our goal is to iteratively create sets $D$ and $K$ until $K$ becomes an MWB while maintaining the invariant that there is a minimum-cost certificate compatible with $D$ and $K$. To this end, we first introduce a series of observations and lemmas that, given a compatible minor $M[D,K]$, allow us to compute a \emph{larger} compatible minor $M[D',K']$ (\Cref{sec:matroid:rules}). Formally, larger means $D' \supseteq D$, $K' \supseteq K$ and $D' \supset D$ or $K' \supset K$. We also refer to these observations and lemmas as \emph{contraction and deletion rules}. We use the phrase \emph{applying a rule} to refer to the process of computing the larger compatible minor $M[D',K']$ and replacing the given one, i.e.~we set $K:=K'$ and $D:=D'$. We will show that exhaustively applying these rules in a certain order until $K$ is a basis and $D = E \setminus K$ yields an MWB that can be verified with minimum-cost.
Afterwards, we show that the rules can be implemented with a polynomial-time algorithm (\Cref{sec:matroid:algorithm}), which will imply the following theorem.

\begin{theorem}
    \label{thm:basis-computation}
    There is a polynomial-time algorithm that computes, for any given weighted uncertainty matroid~$\cM$, an MWB $B$ that can be verified by a minimum-cost certificate.
\end{theorem}

\subsection{Contraction and Deletion Rules}
\label{sec:matroid:rules}

Consider a weighted uncertainty matroid $\cM = (E, \cI, A,w)$.  
Elements that appear in every MWB of $\cM$ or in none can be safely contracted (added to $K$) or deleted (added to $D$). The next observations follow from standard arguments
(proofs are in~\Cref{app:identify}).

\begin{restatable}{obs}{obsUniqueMax}
    \label{obs:matroid:unique:max}
    Let $M':=M[D,K]$ be a compatible minor of $\cM$.  If there is a circuit $C$ in $M[D,K]$ with an element $e \in C$ that has the unique maximum weight in $C$, 
    then $M[D+e,K]$ is a compatible minor of $\cM$. 
\end{restatable}

\begin{restatable}{obs}{obsUniqueMin}
    \label{obs:matroid:unique:min}
    Let $M':=M[D,K]$ be a compatible minor of $\cM$.  If there is a basis $B$ of 
    $M'$ with an element $e\in B$ that has unique minimum weight in $E(M')\setminus \mspan[M']{B - e}$, 
    then $M[D,K+e]$ is a compatible minor of $\cM$.    
\end{restatable}
We continue by giving contraction and deletion rules for \emph{extreme case elements} $e$ satisfying  $w_e = U_e$ or $w_e = L_e$.  The existence of these elements separates our work from~\cite{erlebach_verification}.

\subsubsection{Rules for non-trivial extreme case elements}

We start by considering non-trivial extreme case elements.
For each $w \in \{ w_e \mid e \in E\}$, we define 
$E_w^L = \{ e \in E \mid L_e = w  \land e \text{ is non-trivial}\}$ and $E_w^ U = \{ e \in E \mid U_e = w \land e \text{ is non-trivial} \}$. 

\begin{lemma}
    \label{lem:matroid:cut:non-trivial}
    Let $M':=M[D,K]$ be a compatible minor of  $\cM$.
    Let $e$ be a \emph{non-trivial} element such that (i) $e\in B$ for some MWB $B$ of $M'$,
    (ii) $w_e = U_e$, (iii) $e$ has minimum weight in $E(M')\setminus \mspan[M']{B-e}$  and (iv) $e$ has maximum query cost in $(E(M')\setminus \mspan[M']{B-e}) \cap E_{w_e}^U$. Then $M[D,K+e]$ is a compatible minor of $\cM$.
\end{lemma}

\begin{proof}
    Let $Q^*$ denote a minimum-cost certificate for $\cM$ that is compatible with $M'$. Then, this certificate must also verify some MWB $B'$ for $M'$. 
    
    To prove the lemma, we have to show that there exists a minimum-cost certificate that verifies an MWB $B''$ of $\cM$ with $K + e \subseteq B''$ and $D \cap B'' = \emptyset$. 
    If $e \in B'$, then this directly follows for $Q^*$  and the MWB $B'' = K \cup B'$ of $\cM$. Thus, assume $e \not\in B'$ and let $C$ denote the fundamental circuit of $e$ with respect to~$B'$.
    By \Cref{obs:matroid:cuts}, there exists an $e'\in C-e$ such that $e' \notin \mspan[M']{B-e}$.
    Also, $e' \in B'-e$ because $C'-e \subseteq B'$. By assumption that $e$ has minimum weight in $E(M')\setminus\mspan[M']{B-e}$ and that $B'$ is an MWB of~$M'$, we have $U_e = w_e = w_{e'}$.
    We distinguish the following cases:
    \begin{enumerate}
        \item If $e, e' \in Q^*$, then $Q^*$ also verifies that $B'' = K \cup (B' -e'+e)$ is an MWB by~\Cref{lem:trivial:exchange}. Note that $(B' -e'+e)$ is an MWB for $M'$ and $B''$ is an MWB for $\cM$.
    
        \item If $e \in Q^*$ and $e' \not\in Q^*$, then we must have $w_e \ge U_{e'}$. Otherwise, $Q^*$ would not verify that $e$ has maximum weight in $C$. Since $U_{e'} \ge w_{e'} = w_e = U_e$, we get $U_e = U_{e'}$. 
        By~\Cref{lem:extreme:upperlimit:exchange}, this implies that  $Q = (Q^*  - e + e')$ verifies the MWB $B'' = K \cup  (B' -e'+ e)$.
        Note that $e'$ can be trivial, in which case we can just use $Q = (Q^*  - e)$ instead.

        If $e'$ is trivial, then we clearly have $c(Q) \le c(Q^*)$. Otherwise, i.e., if $e'$ is non-trivial, we have $e' \in E_{w_e}^U$ as we already argued that $U_{e'} = w_{e'} = w_e$. By assumption (iv) of the lemma, this implies $c_e \ge c_{e'}$ and, thus, $c(Q) \le c(Q^*)$.

        \item If $e \not\in Q^*$, then we must have $w_{e'} \le L_e$. Otherwise $Q^*$ cannot verify that $e$ has maximum weight in $C$. However, this implies $w_{e'} = U_e \le L_e$, which can only be the case if $e$ is trivial; a contradiction to the requirements of the lemma.\qedhere
    \end{enumerate}
\end{proof}

The next lemma is a dual of  ~\Cref{lem:matroid:cut:non-trivial} and  can be shown analogously (\Cref{app:identify}). 

\begin{restatable}{lemma}{lemCircuitNontrivial}
    \label{lem:matroid:circuit:non-trivial}
    Let $M':=M[D,K]$ be a compatible minor of $\cM$.
    Let $e$ be a \emph{non-trivial} element such that (i) there exists a circuit $C$ in $M[D,K]$ with $e \in C$, (ii) $w_e = L_e$, (iii) $e$ has maximum weight in $C$ and (iv) $e$ has maximum query cost in $C \cap E_{w_e}^L$. Then $M[D+e,K]$ is a compatible minor of $\cM$.
\end{restatable}

Iteratively and exhaustively applying the deletion rules in~\Cref{lem:matroid:circuit:non-trivial} and~\Cref{obs:matroid:unique:max}, and the contraction rules in~\Cref{lem:matroid:cut:non-trivial} and~\Cref{obs:matroid:unique:min},  yields $D$ and $K$ such that $M[D,K]$ satisfies the following assumption, enabling further rules depending on
this assumption.

\begin{asm}
    \label{asm:matroids:no:non-trivial}
    We may assume that a compatible minor $M':=M[D,K]$ of $\cM$ does not contain non-trivial elements $e$ with any of the following properties:
    \begin{enumerate}
        \item[(i)] $w_e=L_e$ and $e$ is a maximum weight element in a circuit $C$ of $M'$.
        
        \item[(ii)] $w_e = U_e$ and $e$ is a minimum weight element in $E(M')\setminus \mspan[M']{B -e}$ for a basis $B$ of $M'$ that contains $e$.
    \end{enumerate}
\end{asm}

\subsubsection{Rules for trivial extreme case elements}

We continue by introducing deletion and contraction rules for trivial extreme case elements in instances that satisfy~\Cref{asm:matroids:no:non-trivial}. We defer the following two proofs to~\Cref{app:identify}, but remark that they follow a similar proof strategy as the proof of~\Cref{lem:matroid:cut:non-trivial}.

\begin{restatable}{lemma}{lemCutTrivial}
    \label{lem:matroid:cut:trivial}
    Let $M':=M[D,K]$ be a compatible minor of $\cM$ and assume that~\Cref{asm:matroids:no:non-trivial} is satisfied. Let $e$ be a \emph{trivial} element such that (i) there exists basis $B$ of $M'$ with 
    $e\in B$ and (ii) $e$ has minimum weight among elements not in $\mspan[M']{B-e}$. 
    Then $M[D,K+e]$ is a compatible minor of $\cM$.
\end{restatable}

\begin{restatable}{lemma}{lemCircuitTrivial}
    \label{lem:matroid:circuit:trivial}
    Let $M[D,K]$ be a compatible minor of $\cM$ and assume that~\Cref{asm:matroids:no:non-trivial} is satisfied.
    Let $e$ be a \emph{trivial} element such that (i) there exists a circuit $C$ in $M[D,K]$ with $e \in C$ and (ii) $e$ has maximum weight in $C$. Then $M[D+e,K]$ is a compatible minor of $\cM$.
\end{restatable}

Given a compatible minor $M[D,K]$ that satisfies~\Cref{asm:matroids:no:non-trivial},  we check for an element~$e$ that meets the requirements of 
\Cref{lem:matroid:cut:trivial} or~\Cref{lem:matroid:circuit:trivial}. If such an element exists, we apply the corresponding rule to extend $D$ or $K$, 
while maintaining our invariant that $M[D,K]$ is a compatible minor. We then return to applying the rules of the previous sections to restore~\Cref{asm:matroids:no:non-trivial}, 
and repeat this argument until no such element $e$ remains. 
The resulting compatible minor $M[D,K]$ satisfies the following properties:
\begin{enumerate}
    \item $M[D,K]$ does not contain elements $e$ with $w_e=L_e$ such that $e$ is a maximum-weight element in some circuit of $M[D,K]$.
    \item $M[D,K]=:M'$ does not contain elements $e$ with $w_e = U_e$ such that there is a basis $B$ of $M'$ with $e \in B$ and $e$ has minimum weight among elements not in $\mspan[M']{B  - e}$. 
    \item There is a minimum-cost certificate that is compatible with $D$ and $K$.
\end{enumerate}
In any basis $B$ of $M'$ with these properties, no element $e\in B$ with $w_e = U_e$ has minimum-weight in $E(M') \setminus \mspan[M']{B  - e}$. Hence,  such elements $e\in B$ with $w_e = U_e$ cannot be part of any MWB $B$ of $M'$ as we can swap $e$ with an element of $E(M') \setminus \mspan[M']{B  - e}$ that has strictly smaller weight. Thus, we can delete all these elements $e\in B$ with $w_e = U_e$ and add them to $D$ while maintaining the property that there is a minimum-cost certificate that is compatible with $D$ and $K$. 
Note that this corresponds to applying~\Cref{obs:matroid:unique:max}. 
Similarly, we can argue that all remaining elements $e$ with $w_e = L_e$ can be contracted by applying~\Cref{obs:matroid:unique:min}.
This way we compute sets $D$ and $K$ that satisfy the following~assumption.

\begin{asm}
    \label{asm:matroids:no-extrem-elements}
    We may assume that a compatible minor $M[D,K]$ of a weighted uncertainty matroid $\cM$ does not contain elements $e$ with $w_e = U_e$ or $w_e = L_e$. 
\end{asm}

The rules for extreme case elements imply the following (cf.~\Cref{app:identify} for a proof). 

\begin{restatable}{coro}{coroExtremeCaseInstances}
    \label{coro:extremecase:instances}
     Let $M[D,K]$ be a compatible minor of $\cM$ with $A_e = \{L_e,U_e\}$ for all $e \in E(M)$ that satisfies~\Cref{asm:matroids:no:non-trivial}, then $K$ is an MWB that can be verified with a minimum-cost certificate. If in addition the query costs are uniform and all $A_e$ are uniform and non-trivial, then every MWB $B$ of $M$ has the same verification cost.
\end{restatable}

\subsubsection{Rules for instances without extreme case elements}

Instances $M':=M[D,K]$ that satisfy~\Cref{asm:matroids:no-extrem-elements} essentially behave like 
instances 
with 
open intervals, since all elements 
satisfy $w_e \in (L_e,U_e)$ (although we still might have $A_e \not= (L_e,U_e)$). 
For the special case of graphic matroids with open intervals, the results 
in~\cite{ErlebachLMS22,erlebach_verification} imply that \emph{all} MWB can be verified with minimum-cost certificates. 
In~\Cref{app:identify}, we generalize these insights to arbitrary matroids and 
prove the following lemma.

\begin{restatable}{lemma}{lemUniqueMWB}
    \label{lem:unique:MWB}
    Let $M':=M[D,K]$ be a compatible minor of  $\cM$ such that~\Cref{asm:matroids:no-extrem-elements} is satisfied, and let $B$ be any MWB of $M'$. Then, $M[D',K']$ with $K' = K \cup B$ and $D' = D \cup (E(M')\setminus B)$ is a compatible minor of $\cM$.
\end{restatable}

Once we create a compatible minor $M':=M[D,K]$ that satisfies~\Cref{asm:matroids:no-extrem-elements} by applying the contraction and deletion rules of the previous sections,~\Cref{lem:unique:MWB} allows us to just fix any MWB $B$ of $M'$, contract (delete) the elements of $B$ (of $E(M')\setminus B$) and add them to $K$ (to $D$). Afterwards, $K$ is a basis that can be verified with a minimum-cost certificate.

\subsection{Polynomial-time Algorithm}
\label{sec:matroid:algorithm}

The results of \Cref{{sec:matroid:rules}} imply a sequence of contraction and deletion rules applied to an uncertainty matroid, leading to a compatible minor that has a unique MWB with a minimum-cost certificate. ~\Cref{alg:basis} exhaustively applies these rules in the right order to compute this MWB, and runs
 in polynomial time (cf.~\Cref{app:identify}). This proves~\Cref{thm:basis-computation}.

\begin{algorithm}[htb]
\caption{Computing an MWB $B$ that can be verified with a minimum-cost certificate}
\begin{algorithmic}[1]
    \State {\bf Input:} A weighted uncertainty matroid $\cM=(E,\cI,A,w)$ 
    \State $K, D \gets \emptyset$
    \Repeat
        \State $M' \gets M[D,K]$;
        \State Select the element
        $e \in E(M')$ that satisfies the first condition listed, in the following order, and add to $D$ or $K$ as specified:
        \IIf{\( e \) satisfies \Cref{obs:matroid:unique:max}}{ add \( e \) to \( D \); }
        \IIf{\( e \) satisfies \Cref{obs:matroid:unique:min}}{ add \( e \) to \( K \); }
        \IIf{\( e \) satisfies \Cref{lem:matroid:cut:non-trivial}}{ add \( e \) to \( K \); }
        \IIf{\( e \) satisfies \Cref{lem:matroid:circuit:non-trivial}}{ add \( e \) to \( D \); }
        \IIf{\( e \) satisfies \Cref{lem:matroid:cut:trivial}}{ add \( e \) to \( K \); }
        \IIf{\( e \) satisfies \Cref{lem:matroid:circuit:trivial}}{ add \( e \) to \( D \); }
    \Until{none of the conditions in lines 6--11 apply} 
    
    \State Compute an MWB $B$ of $M'=M[D,K]$. Add $B$ to $K$ and $E(M')\setminus B$ to D;
    \State \Return $K$;
\end{algorithmic}
\label{alg:basis}
\end{algorithm}

\section{Computing a minimum-cost certificate to verify a given 
basis} \label{sec: construct cert}
In this section, we provide an algorithm for computing a minimum-cost certificate that verifies a given MWB.\ In particular, the algorithm can be applied to the MWB produced by~\Cref{alg:basis}, thereby yielding a minimum-cost certificate for $\cM$.
Lemma \ref{lemma: cert contains for Le Ue} is a key result that dictates specifically what elements must be in $Q$ in order for it to adhere to Corollary \ref{coro:certificate:characterization:circuits}.

\begin{restatable}{lemma}{lemmaInCert}
    \label{lemma: cert contains for Le Ue}
    Let $B$ be an MWB of a weighted uncertainty matroid $\mathcal{M} = (E, \mathcal{I}, A,w)$. For a non-basis element $e$ of $B$, let $C_e$ be the fundamental circuit of $e$ with respect to $B$, let $F_e$ be the set of elements $f \in C_e - e$ with $U_f > L_e$, and let $\hat{F_e}$ be the set of elements $f \in C_e - e$ with $U_f > w_e$. 
    A subset $Q \subseteq E$ is a certificate of $B$ if and only if for each non-basis element $e$ of $B$, the following holds:
    \begin{enumerate}
        \item If $w_e \geq U_f$ for all $f \in C_e-e$, and there exists $f' \in C_e - e$ such that $w_{f'} > L_e$, then $e \in Q$.
        \item If $w_e \geq U_f$  for all $f \in C_e-e$, and $w_f \leq L_e$ for all $f \in C_e-e$, then $e \in Q$ or $F_e\subseteq Q$.
        \item If there exists $f \in C_e - e$ such that $w_e < U_f$, and there exists $f' \in C_e-e$ (could have $f' = f$) such that $w_{f'} > L_e$, then $\hat{F_e} + e\subseteq Q$.
        \item If there exists $f \in C_e-e$ such that $w_e < U_{f}$, and for all $f' \in C_e -e$ it holds that $w_{f'} \leq L_e$, then $F_e\subseteq Q$ or $\hat{F_e} + e\subseteq Q$. 
    \end{enumerate}
\end{restatable}

    In all cases, the construction of $Q$ verifies that each element $e \notin B$ is a maximum-weight element on the unique circuit it forms in $B$, which completes the proof by \Cref{coro:certificate:characterization:circuits}. 

    Intuitively,~\Cref{lemma: cert contains for Le Ue} exhaustively formulates the \enquote{choices} that a certificate $Q$ has to make to ensure that~\Cref{coro:certificate:characterization:circuits} is satisfied for a fundamental cycle $C_e$ with respect to $B$ of an $e \not\in B$. For example, if $C_e$ falls into the second case of the lemma, then we must have $F_e \subseteq Q$ or $e \in Q$ (or both). The \enquote{choices} for different $e,e' \not\in B$ however, are not necessarily disjoint, e.g., if we have $F_e \cap F_{e'} \not= \emptyset$. Similar to verification algorithms in the literature, we define an auxiliary graph $G^B$ such that every vertex cover of $G^B$ corresponds to a combination of \enquote{choices} of~\Cref{lemma: cert contains for Le Ue}, and vice versa. The following definition makes this more precise.

  \begin{definition}
        \label{def:auxiliary_graph}
        Given a weighted uncertainty matroid $\cM = (E, \mathcal{I},A, w)$ and an MWB $B$, we define the bipartite (with the exception of self-loops) auxiliary graph $G^B = (V^B, E^B)$ with $V^B= E$ and $E^B = \bigcup_{e \in E\setminus B} E^B_e$, where the sets $E_e^B$ for $e \in E\setminus B$ are defined as follows:
    \begin{enumerate}
        \item If $w_e \geq U_f$ for all $f \in C_e-e$, and there is a $f' \in C_e - e$ 
        with $w_{f'} > L_e$, then~$E^B_e = \{\{e,e\}\}$.
        \item If $w_e \geq U_f$  for all $f \in C_e-e$, and $w_f \leq L_e$ for all $f \in C_e-e$, then $E^B_e = \{ \{e,f\} \mid f \in F_e\}$.
        \item If there exists $f \in C_e - e$ such that $w_e < U_f$, and there exists $f' \in C_e-e$ (could have $f' = f$) such that $w_{f'} > L_e$, then $E^B_e=\{\{f,f\} \mid f \in \hat{F_e} + e\}$.
        \item If there exists $f \in C_e-e$ such that $w_e < U_{f}$, and for all $f' \in C_e -e$ it holds that $w_{f'} \leq L_e$, then
        $E_e^B = \{\{e,f\} \mid f \in F_e \setminus \hat{F}_e\} \cup \{\{f,f\} \mid f \in \hat{F}_e\}$.
    \end{enumerate}
 \end{definition}

Each case of~\Cref{def:auxiliary_graph} corresponds to a case of~\Cref{lemma: cert contains for Le Ue} and exactly models the corresponding \enquote{choice} of the lemma. The following corollary 
follows from~\Cref{lemma: cert contains for Le Ue} and formally proves (\Cref{app:construct}) the connection between~\Cref{def:auxiliary_graph} and~\Cref{lemma: cert contains for Le Ue}. 

 \begin{restatable}{coro}{coroVertexCover}
    \label{coro:vertex:cover}
    Given a weighted uncertainty matroid $\mathcal{M} = (E, \mathcal{I}, A,w)$ with MWB $B$, a set~$Q$ is a certificate of $B$ if and only if $Q$ is a vertex cover of $G^B$.
 \end{restatable}
 
Since a set $Q\subseteq E$ is a certificate for $B$ if and only if $Q$ is a vertex cover of $G^B$, we can compute a minimum-cost certificate for $B$ by computing a minimum-cost vertex cover for $G^B$, using the query costs $c_e$ as weights for the vertices $V^B = E$. Note that $G^B$ is bipartite except for the self-loop edges. 
To employ polynomial-time vertex-cover algorithms for bipartite graphs~(see, e.g.,~\cite{Schrijver2003book}), we modify $G^B$ by replacing each loop-edge $\{e,e\} \in E^B$ with an edge $\{e,v_e\}$ to a distinct new dummy vertex $v_e$ with weight $w(v_e) = w_e + \epsilon$ for some $\epsilon > 0$. Executing the algorithm, formalized as Algorithm~\ref{alg:alternative}, for the MWB $B$ computed by~\Cref{alg:basis} yields a polynomial-time algorithm that computes a minimum-cost certificate.

\begin{algorithm}[t]
\caption{Identifying a Minimum-Cost Certificate of $B$}
\begin{algorithmic}[1]
    \Require $\mathcal{M} = (E, \mathcal{I}, A,w)$  is a weighted uncertainty matroid.  
    \State Let $B$ be an MWB of $\cM$ as computed by~\Cref{alg:basis}
    \State Let $G^B=(V^B,E^B)$ be the auxiliary graph as defined in~\Cref{def:auxiliary_graph}
    \State Let $S^*$ be a minimum-weight vertex cover of $G^B$ for the weights $w(e) = c_{e}$, $e \in V^B=E$.
    \State \Return $S^*$ and $B$
\end{algorithmic}
\label{alg:alternative}
\end{algorithm}

\begin{restatable}{theorem}{thmFinalAlg}
        Let $\cM  = (E, \mathcal{I}, A,w)$ be a weighted uncertainty matroid. 
        Algorithm~\ref{alg:alternative} is a polynomial-time algorithm that computes a minimum-cost certificate $Q$ and an MWB $B$ of $\cM$ that is verified by $Q$.
\end{restatable}

\section{Applications} \label{sec: applications}

In this section, we show that our verification algorithm and our structural insights, in particular~\Cref{lemma: cert contains for Le Ue}, can be used to obtain new results for the adaptive online variant and for two learning-augmented variants of the MWB problem under uncertainty. 

\subsection{A best-possible adaptive online algorithm for a given MWB}
\label{sec:fixed:trees}

In the following, we give an optimal algorithm for the adaptive online problem under the assumption that we are given an MWB $B$ that can be verified with a minimum-cost certificate. 

In the adaptive online problem, the weights $w_e$ are initially unknown, and an algorithm has access only to the uncertainty matroid $\cM$. The goal of the algorithm is to adaptively query elements until the set of queried elements $Q$ verifies some MWB $B$ of $\cM$ (cf.~\Cref{def: what it means for S to verify T}).
Algorithms are analyzed in terms of their \emph{competitive ratio}. An algorithm is \emph{$\rho$-competitive} if $c(Q(\cM)) \le \rho \cdot c(Q^*(\cM))$ for all uncertainty matroids $\cM$, where $Q(\cM)$
is the set of elements queried by the algorithm on $\cM$,
and $Q^*(\cM)$ is a minimum-cost certificate for $\cM$. The competitive ratio of an algorithm is the minimum $\rho$ such that the algorithm is $\rho$-competitive.

The best-possible competitive ratio
heavily depends on the type of uncertainty areas. If all uncertainty areas are either open, i.e., $A_e = (L_e,U_e)$, or trivial, i.e., $A_e = \{w_e\}$, then the best-possible competitive ratio is $2$~\cite{erlebach08steiner_uncertainty}. Once the uncertainty areas can be closed intervals, the best-possible competitive ratio increases to~$n$~\cite[Section~7]{GuptaSS16}. Both results hold for uniform query costs ($c_e = 1$ for $e \in E$).
Here, we consider the case of uniform query costs.

It is not hard to see that the lower bound of $2$ for minimum spanning trees in the case of open uncertainty areas~\cite{erlebach08steiner_uncertainty} holds even if the algorithm is given an MWB $B$ with the promise that (i) $B$ is indeed an MWB for the unknown weights and (ii) $B$ can be verified with a minimum-cost certificate. Note that, for open uncertainty areas, (ii) is an implication of~(i). 

\begin{observation}[Follows from \cite{erlebach08steiner_uncertainty}]
    No deterministic algorithm for the online adaptive MWB problem with uniform query costs is better than $2$-competitive, even if the algorithm is given an MWB $B$ w.r.t.~to the unknown weights.
\end{observation}


In contrast, the lower bound of $n$ for the case of general uncertainty areas~\cite[Section~7]{GuptaSS16} 
does not hold if the algorithm receives the same promise as described above. This contrast highlights that the increase in the competitive ratio from {\em open} to {\em general} uncertainty areas stems from the algorithm's task of \emph{finding} 
rather than merely \emph{verifying} an MWB $B$. In the following, we formally prove this insight. The core of the proof is an application of~\Cref{lemma: cert contains for Le Ue}.

\begin{lemma}
\label{lem:comp:knowntree}
    There exists an online adaptive algorithm that,  given an uncertainty matroid $\cM$ and an MWB $B$ (w.r.t.~the unknown weights $w$) that can be verified with a minimum-cardinality certificate $Q^*$, computes a certificate $Q$ with $|Q|\le 2 |Q^*|$ that verifies $B$.
\end{lemma}

\begin{proof} To prove the lemma, we define the following algorithm:
    \begin{enumerate}
        \item Initialize $Q = \emptyset$.
        \item For each $e \in E \setminus B$:
        \begin{enumerate}
            \item Let $C_e$ denote the fundamental circuit of $e$ with respect to $B$.
            \item While there exists an $f \in C_e - e$ with $U_f(Q) > L_e(Q)$:
            \begin{enumerate}
                \item If there is an element $g \in C_e \setminus (\{e\} \cup Q)$ with $U_{g}(Q) > L_e(Q)$, then let $\hat{g}$ denote such an element with maximum $U_{\hat{g}}(Q)$. Add $\hat{g}$ to $Q$ and query $\hat{g}$.
                \item If $e \not\in Q$, then add $e$ to $Q$ and query $e$.
            \end{enumerate}
        \end{enumerate}
        \item Return $Q$.
    \end{enumerate}

    It is not hard to see that the set $Q$ computed by the algorithm is a certificate that verifies~$B$: Exploiting the assumption that $B$ is indeed an MWB, the definition of Step 2b ensures that $Q$ satisfies~\Cref{coro:certificate:characterization:circuits}. Hence, $Q$ is a certificate that verifies $B$.

    It remains to show that $|Q| \le 2 \cdot |Q^*|$. The proof will heavily exploit that $Q^*$, by assumption, is a certificate that verifies $B$. In particular, this means that $Q^*$ satisfies~\Cref{lemma: cert contains for Le Ue} for $B$, which will help us to bound $|Q|$ in terms of $|Q^*|$.

    Let $E\setminus B = \{e_1, \ldots, e_k\}$ be indexed in the order in which the algorithm considers the elements of $E \setminus B$ in the loop of Step 2. For each $j \in \{1,\ldots, k\}$ let $Q_j$ denote the set of elements that the algorithm queries during iteration $j$ of the loop of Step 2. We show that 
    \begin{equation}
    \label{eq:online:1}
        \frac{|Q_j|}{|Q_j \cap Q^*|} \le 2
    \end{equation}
    holds for each $Q_j$ with $j \in \{1,\ldots, k\}$. Since the sets $Q_j$ form a partition of $Q$, this then implies $|Q| \le 2 |Q^*|$.

    Fix an arbitrary $Q_j$. Define $P_j = \bigcup_{j' < j} Q_{j'}$.
    We show that $Q_j$ satisfies~\eqref{eq:online:1} via case distinction over the cardinality of $Q_j$.
    \begin{enumerate}
        \item If $|Q_j| = 0$, then~\eqref{eq:online:1} holds trivially.
        \item If $|Q_j| = 1$, then the definition of the loop of Step 2b implies that we must have $Q_j = \{e_j\}$ as $e_j \in Q_j$ holds by Step 2bii of the algorithm. The fact that the loop of Step 2b is executed but only queries $e_j$ implies that all $g \in C_{e_j} \setminus (\{e_j\} \cup P_j)$ have $w_g \le U_g(P_j) \le L_{e_j}(P_j) = L_{e_j}$ (cf.~Step 2bi). However, then the fact that the loop of Step 2b is executed implies that there is some $g \in C_{e_j} \cap P_j$ with $U_g(P_j) = w_g > L_{e_j}(P_j) = L_{e_j}$. This in turn implies that $C_{e_j}$ either satisfies the first or third case of~\Cref{lemma: cert contains for Le Ue}. Hence, $e_j \in Q^*$ and $\frac{|Q_j|}{|Q_j \cap Q^*|}=1$.
        \item If $|Q_j| = 2$, then~\Cref{lemma: cert contains for Le Ue} implies $|Q_j \cap Q^*| \ge 1$. Hence, $\frac{|Q_j|}{|Q_j \cap Q^*|} \le 2$
        \item If $|Q_j| > 2$, then the loop of Step 2b is executed at least twice. 
        For the two elements that are queried in the first iteration, we can argue that $Q^*$ has to contain at least one of them as in the previous case.
        
        Let $Q_j'$ denote the subset of $Q_j$ that is queried after the first iteration of this Step 2b loop. All $g \in Q_{j}'$ have to satisfy $U_g > w_{e_j}$, as otherwise the loop would terminate before $g$ is queried. This implies that all $g \in Q_j'$ are in $\hat{F}_{e_j}$ (defined as in~\Cref{lemma: cert contains for Le Ue}). Hence, $C_{e_j}$ either satisfies Case 3 or 4 of~\Cref{lemma: cert contains for Le Ue} and, thus, $Q_j' \subseteq \hat{F}_{e_j} \subseteq Q^*$. We can conclude with $\frac{|Q_j|}{|Q_j \cap Q^*|}\le \frac{2+ |Q_j'|}{1+ |Q_j'|} \le 2$.\qedhere 
    \end{enumerate}    
\end{proof}

\subsection{Learning-augmented algorithms}

Recently, online adaptive problems under explorable uncertainty have been studied in \emph{learning-augmented} settings~\cite{ErlebachLMS22,ErlebachLMS23}, where the algorithm has access to imperfect predictions $\hat{w}_e$ on the unknown weights $w_e$. The goal is to design algorithms that leverage the access to the predictions to achieve an improved competitive ratio in case the predictions are accurate, and at the same time maintain worst-case guarantees even if the predictions are arbitrarily wrong. Such learning-augmented algorithms are analyzed w.r.t.~their consistency, the competitive ratio if $w_e = \hat{w}_e$ for all $e \in E$, and robustness, the competitive ratio for arbitrary predictions~\cite{LykourisV21,PurohitSK18}. 

For the online adaptive MST problem with open uncertainty areas,~\cite{ErlebachLMS22} gave a $1.5$-consistent and $2$-robust algorithm. In the following we show that our verification algorithm and the insights from the previous section can be used to design learning-augmented algorithms for the online adaptive MWB problem with general uncertainty areas for two prediction models:
\begin{enumerate}
    \item \textbf{Weight predictions:} As in~\cite{ErlebachLMS22,ErlebachLMS23}, the algorithm has access to predictions $\hat{w}_e$ on the unknown element weights $w_e$.
    \item \textbf{Basis predictions:} The algorithm has access to a prediction $\hat{B}$ on the MWB $B^*$ that can be verified with a minimum-cardinality certificate.
\end{enumerate}

\subsubsection{Weight predictions}

Our verification algorithm implies a  $1$-consistent and $n$-robust algorithm for the case of uniform query costs. Note that this consistency and robustness are optimal due to the aforementioned lower bound given in~\cite{GuptaSS16}.
\begin{enumerate}
    \item 
    Use the verification algorithm to compute and query a minimum-cost certificate under the assumption that the predictions are accurate, i.e., $w_e = \hat{w}_e$ for all $e \in E$.
    \item If the instance is not solved yet,
    then query all remaining elements. 
\end{enumerate}

\begin{restatable}{observation}{obsNaive}
\label{obs:naive}
    The algorithm above runs in polynomial  time and is $1$-consistent and $n$-robust for the online adaptive MWB problem with weight predictions and uniform query~costs.
\end{restatable}

This naive algorithm is a first step towards \emph{smooth} learning augmented-algorithms for the online adaptive MWB problem with weight predictions, i.e., algorithms with a competitive ratio that smoothly degrades from $1$ to $n$ depending on some error measure describing the quality of the predictions. Next, we give such a smooth algorithm (with a worse consistency) for basis predictions. Since weight predictions can be used to compute a basis prediction, the algorithm can also be interpreted as a smooth algorithm for weight predictions.

\subsubsection{Basis predictions}
\label{sec:basis:preds}

We consider the prediction model where algorithms have access to a prediction $\hat{B}$ on an MWB $B^*$ that can be verified with a minimum cardinality query set $Q^*$. 
W.l.o.g.\ we may assume that $\hat{B}$ is indeed a basis of the given uncertainty matroid. This is easy to check and if $\hat{B}$ is not a basis, we remove a minimum number of elements to achieve independence and then augment it to a basis.

We say that the prediction $\hat{B}$ is \emph{correct} if it is indeed an MWB that can be verified with some minimum cardinality query set. Otherwise, the prediction is \emph{incorrect}. We now define error measures to quantify the quality of a prediction $\hat{B}$.
Note that a prediction $\hat{B}$, can have two different types of error. First, $\hat{B}$ might not actually be an MWB. Second, even if $\hat{B}$ is an MWB, it might not be an MWB of minimum verification cost. Our error measure takes both of these error types into account.
To this end, let $C_e$ for $e \in E \setminus \hat{B}$ denote the fundamental circuit of $e$ with respect to $\hat{B}$. Define such a circuit $C_e$ to be \emph{correct} if $w_e \ge w_f$ for all $f \in C_e - e$ and $w_f \le w_{g}$ for all $f \in C_e-e$ and $g \in E\setminus\mspan[M]{\hat{B}-f}$. Otherwise, call $C_e$ \emph{incorrect}.
Let $\hat{B}_s \subseteq \hat{B}$ denote the set of elements in $\hat{B}$ that are part of a correct fundamental circuit $C_e$, $e \in E \setminus \hat{B}$. 
We define two errors:
\begin{itemize}
    \item Let $\hat{B}_s^*$ 
    denote an MWB with $\hat{B}_s \subseteq \hat{B}_s^*$ such that $\hat{B}_s^*$ has minimum verification cost among all MWB's that contain $\hat{B}_s$ and let $Q'$ denote a minimum-cardinality certificate for~$\hat{B}_s^*$. 
    In~\Cref{app:applications}, we argue that such an MWB $\hat{B}_s^*$ always exists. Define $\eta_1 = |Q'| - |Q^*|$, where $Q^*$ is the minimum-cardinality certificate. Note that if $\hat{B}$ is an MWB, then $\eta_1$ is just the difference between the verification cost of $\hat{B}$ and the minimum verification cost~$|Q^*|$.
    \item Let $\eta_2$ denote the number of incorrect circuits $C_e$ with $e \in E \setminus \hat{B}$.
\end{itemize}

Intuitively, $\eta_1$ measures how much larger the verification cost of $\hat{B}$ is compared to the minimum-cardinality query set. Error $\eta_2$  captures how far $\hat{B}$ is from actually being an MWB.

The following theorem can be achieved using a slight variation of the algorithm 
in~\Cref{sec:fixed:trees}, taking into account that elements $e \in E \setminus \hat{B}$ might not be maximal on circuit $C_e$. In such cases, the algorithm deviates from the procedure in~\Cref{sec:fixed:trees} by querying all elements in circuit $C_e$. In~\Cref{app:applications},  we give the formal definition of this algorithm and the proof of the theorem, and argue that the factor $c_{\max}$ of the error $\eta_2$ is indeed necessary.

\begin{restatable}{theorem}{thmBasisPredictions}
\label{thm:basis:predictions}
Given an instance of the online adaptive MWB problem with a predicted basis~$\hat{B}$ and uniform query costs, there is an algorithm that computes a certificate $Q$ for some MWB $B$ and satisfies
$|Q| \le \min\{ 2 \cdot (|Q^*| + \eta_1) + \eta_2 \cdot c_{\max}, n\}$, where $Q^*$ is a minimum-cardinality certificate and $c_{\max}$ is the size of the largest circuit.
In particular, the algorithm is $2$-consistent and $n$-robust.
\end{restatable}
	
\bibliography{main}	

\newpage 
\appendix

\section{The power of adaptivity}
\label{app:further:discussion}

Given an uncertainty matroid $\mathcal{M} = (E, \mathcal{I}, A)$, we say that a set $Q \subseteq E$ is a \emph{non-adaptive} certificate if, for every weight vector $w'$ with $w'_e \in A_e$ for all $e \in E$, $Q$ is a certificate for the weighted uncertainty matroid $\mathcal{M}' = (E, \mathcal{I}, A, w')$. Note that Merino and Soto~\cite{MerinoS19} consider the problem of computing a non-adaptive certificate of minimum query cost for a given uncertainty matroid $\mathcal{M} = (E, \mathcal{I}, A)$. We can observe that the gap between the cost of a minimum-cost non-adaptive certificate for $\mathcal{M} = (E, \mathcal{I}, A)$ can be arbitrarily large compared to the cost of a minimum-cost certificate for the weighted uncertainty matroid $\mathcal{M}' = (E, \mathcal{I}, A, w)$ for a single weight vector $w$.

\begin{observation}
    Fix any $\rho \ge 1$. There exists an uncertainty matroid $\mathcal{M} = (E, \mathcal{I}, A)$, a weight vector $w$ and a query cost function $c$ such that $$\frac{c(Q^*_N)}{c(Q^*)} \ge \frac{\sum_{e \in E} c_e}{\min_{e \in E} c_e} \ge \rho,$$ where $Q^*_N$ is a minimum-cost non-adaptive certificate for $\mathcal{M} = (E, \mathcal{I}, A)$ and $Q^*$ is a minimum-cost certificate for  $\mathcal{M}' = (E, \mathcal{I}, A, w)$. For $\rho = n$, the statement holds for uniform query costs.
\end{observation}

\begin{proof}
    Consider the uniform matroid $E = (E, \mathcal{I})$ with cardinality constraint $k = 1$, that is, $\mathcal{I} = \{ S \subseteq E \colon |S| \leq 1\}$. Let $E = \{e_1, \ldots, e_n\}$. For the associated sets we define $A_{e_1} = [0,2]$ and $A_{e_i} = [1,3]$ for all $i \neq 1$. The costs are given by $c_{e_1} = 1$, while for the remaining elements we set $c_{e_i} = \tfrac{\rho - 1}{n-1}$ whenever $i \neq 1$.
    
    The only non-adaptive certificate for $\mathcal{M} = (E, \mathcal{I}, A)$ is $Q^*_N = E$, which implies $c(Q^*_N) = \rho$.

    Consider the weights $w$ with $w_{e_1} = 0$ and $w_{e_i} = 2$ for all $i\not=1$. Then, $Q^* = \{e_1\}$ is the minimum-cost certificate of $\mathcal{M}' = (E, \mathcal{I}, A, w)$ and verifies the MWB $B = \{e_1\}$.

    We can conclude with $\frac{c(Q^*_N)}{c(Q^*)} = \rho$. Note that if $\rho = n$, then $c_e = 1$ for all $e \in E$.
\end{proof}

\section{Missing proofs from \Cref{sec: prelim}}
\label{app:prelim}

\ObsMatroidCuts*

\begin{proof}
    For the sake of contradiction, assume that any $e'\in C-e$ is in $\mspan[M]{B-e}$. 
    Consider the sets~$B$ and $C-e$, both of which are independent sets. The set  $B$ is independent because it is a basis, while $C - e$  is independent since $C$ is a circuit and thus minimally dependent.
    Since $C-e$ is independent and $C-e \subseteq \mspan[M]{B-e}$, we have $ |C-e| = r(C-e) \le r(B-e) < |B|$. By the augmentation property of matroids, it must be possible to augment $C - e$ to a basis $B'$ by adding only elements of $B$. Since~$C$ is dependent, we must have $e \not\in B'$. However, then we have $B' \subseteq \mspan[M]{B - e}$ since $B - e \subseteq \mspan[M]{B - e}$ and $C - e \subseteq \mspan[M]{B - e}$. Thus, $r(B') \le r(B - e) < r(E)$, which is a contradiction to $B'$ being a basis.
\end{proof}

\coroCharacterization*

\begin{proof}
    We show that $U_e(Q) \le L_f(Q)$ for all $e \in B$ and $f \in (E\setminus \mspan[M]{B-e})-e$ holds if and only if $U_e(Q) \le L_f(Q)$ for all $f \not\in B$ and $e \in C_f-f$. Then,~\Cref{lem:certificate:characterization:cuts} implies the corollary.

    First, assume that $U_e(Q) \le L_f(Q)$ for all $e \in B$ and $f \in (E\setminus \mspan[M]{B-e})-e$. Consider an arbitrary $f \not\in B$ and $e \in C_f - f$. Then, $f \in (E\setminus \mspan[M]{B-e})-e$. Our assumption implies $U_e(Q) \le L_f(Q)$. Since this holds for arbitrary $f \not\in B$ and $e \in C_f - f$, we can conclude that  $U_e(Q) \le L_f(Q)$ for all $f \not\in B$ and $e \in C_f-f$.

    Next, assume that $U_e(Q) \le L_f(Q)$ for all $f \not\in B$ and $e \in C_f-f$. Consider an arbitrary $e \in B$ and $f \in (E\setminus \mspan[M]{B-e})-e$. Then, $e \in C_f - f$. Our assumption implies $U_e(Q) \le L_f(Q)$. 
    Since this holds for arbitrary $e \in B$ and $f \in (E\setminus \mspan[M]{B-e})-e$, we can conclude that  $U_e(Q) \le L_f(Q)$ for all $e \in B$ and $f \in (E\setminus \mspan[M]{B-e})-e$.
\end{proof}

\lemLowerLimitExchange*

\begin{proof}
    During this proof, we use $X_f := E \setminus \mspan[M]{B-f}$ and $X'_f := E \setminus \mspan[M]{B'-f}$.

    To show that $Q'$ is a certificate for $B'$, we show for each $f \in B'$ that $U_f(Q') \le L_{f'}(Q')$ holds for each $f' \in X'_f - f$. By~\Cref{lem:certificate:characterization:cuts}, this implies that  $Q'$ is a certificate for $B'$.
    We distinguish between the two cases (1) $f=e$ and (2) $f \not= e$.

    \textbf{Case (1):} Assume $f = e$ and consider $X'_f = X'_e$. Since $B' = B + e - e'$, we have $X'_e = X_{e'}$. By assumption that $Q$ is a certificate for $B$ with $e' \in B$, we must have $w_{e'} \le U_{e'}(Q) \le L_{f'}(Q)$ for all $f' \in X_{e'} - e'$ by~\Cref{lem:certificate:characterization:cuts}.
    Since $X'_e = X_{e'}$, $e \in Q'$ and $w_e = w_{e'}=L_{e'}$, this implies $w_e = U_e(Q') \le L_{f'}(Q')$ for all $f' \in X_{e'}$.

    \textbf{Case (2):} Assume $f \not= e$ and consider  $X'_f$. We distinguish between $X'_f - f \subseteq E\setminus B$ and $X'_f - f \not\subseteq E\setminus B$.

    \begin{itemize}
        \item If $X'_f - f \subseteq E\setminus B$, then $X_f = X'_f$ and $e,e' \not\in X'_f$. In this case, $U_f(Q') = U_f(Q) \le L_{f'}(Q) = L_{f'}(Q')$  holds for all $f' \in X'_f-f$ by definition of $Q'$ and by our assumption that $Q$ is a certificate for $B$ and~\Cref{lem:certificate:characterization:cuts}.
        \item If  $X_f' -f \not\subseteq E\setminus B$, then we must have $e' \in X_f'$. 
        Consequently, $f \in C_{e} = C'_{e'}$, where $C'_{e'}$ is the fundamental circuit of $e$ with respect to $B'$. 
        Note that this also implies $e \in X_f$.
        Since $Q$ is a certificate for $B$ with $e\not\in B$, we must have $w_{e'} = L_{e'}(Q') = L_{e}(Q) \ge U_f(Q) = U_f(Q')$ by~\Cref{lem:certificate:characterization:cuts}. Note that $w_{e'} = L_{e'}(Q') = L_{e}(Q)$ holds by definition because of $w_e=w_{e'}=L_e=L_{e'}$.

        Next, consider an arbitrary $f' \in X_f' - f - e'$ and the circuit $C_{f'}$. 
        Since $C_{f'}$ is dependent but $C_{f'} - f'$ is independent, the circuit $C_{f'}$ must contain at least one element $g \not\in \mspan[M]{B'-f}-f'$. Otherwise $C_{f'} - f' \subseteq B'-f$ and $f' \in X_f'$ would contradict $C_{f'}$ being dependent.
        The only members of $C_{f'} - f' \subseteq B - f'$ that are not in $\mspan[M]{B'-f}-f'$ are $e'$ and $f$.
        Thus, $e' \in C_{f'}$ or $f \in C_{f'}$.

        If $f \in C_{f'}$, then the definition of $Q'$ and our assumption that $Q$ is a certificate for $B$ with $f \in B$ and $f' \not\in B$ imply $U_f(Q') = U_f(Q) \le L_{f'}(Q) = L_{f'}(Q')$.

        If $e' \in C_{f'}$, then the definition of $Q'$ and our assumption that $Q$ is a certificate for $B$ with $e' \in B$ and $f' \not\in B$ imply $w_{e'} \le  U_{e'}(Q) \le L_{f'}(Q) = L_{f'}(Q')$. As we already argued above that $U_f(Q') \le w_{e'}$, we can conclude with $U_f(Q') \le L_{f'}(Q')$. \qedhere
    \end{itemize}
\end{proof}

\lemTrivialExchange*

\begin{proof}
    Consider an arbitrary weight assignment $w'$ that is consistent with $Q$. By assumption $w'_e= w'_{e'}=w_e= w_{e'}$ . Since $Q$ is a certificate that verifies $B$, we have that $B$ is an MWB for the weights $w'$.  Using that $B' = B -e + e'$ is independent by assumption and $w'_e= w'_{e'}$, this implies that $B'$ is an MWB for weights $w'$. As this argument holds for every weight assignment that is consistent with $Q$, we can conclude that $Q$ verifies $B'$.
\end{proof}

\section{Missing proofs from \Cref{sec: identify MWB}}
\label{app:identify}

\obsUniqueMin*

\begin{proof}
    We show that $e$ is not part of any MWB for $M'$, which implies that $M[D+e,K]$ is a compatible minor.    For the sake of contradiction, assume that there is an MWB $B$ of $M[D,K]=:M'$ with $e \in B$.
    Since $C$ is a circuit with $e \in C$, 
    there is an element $e' \in C-e$ which is not in  $\mspan[M']{B-e}$ due to \Cref{obs:matroid:cuts}. By assumption, we have $w_{e'} < w_e$ and $w(B') < w(B)$ for $B' = B - e + e'$.
    Furthermore, $B'$ is a basis by choice of $e' \not\in \mspan[M']{B-e}$ and therefore $w(B') < w(B)$ is a contradiction to $B$ being an MWB of $M'$. Hence, $e$ is not part of any MWB of $M[D,K]$ and thus of $\cM$.  
\end{proof}

\obsUniqueMax*

\begin{proof}
    Let $B$ be a basis of 
    $M'$ with $e \in B$ having the unique minimum weight in $E(M')\setminus \mspan[M']{B -e}$. 
    We show that $e$ is part of every MWB for $M'$, which implies that $M[D,K+e]$ is a compatible minor.
    For the sake of contradiction, assume that there is an MWB $B'$ of $M'$ with $e \not\in B'$.

    Let $C$ be the fundamental circuit of $e$ with respect to $B'$.
    By \Cref{obs:matroid:cuts}, $C-e$ contains an element $e'$ with $e' \notin \mspan[M']{B -e}$. By assumption, we have $w_e < w_{e'}$. However, this implies $w(B') > w(B'')$ for the basis $B'' = B' - e' + e$, a contradiction to $B'$ being an MWB of $M'$. Hence, any MWB contains $e$ and  $M[D,K+e]$ is a compatible minor of $\cM$.\qedhere
\end{proof}

\lemCircuitNontrivial*

\begin{proof}
    Let $Q^*$ denote a minimum-cost certificate that is compatible with $M'$. Then, this certificate must also verify some MWB $B'$ for $M'$.

    To prove the lemma, we have to show that there exists a minimum-cost certificate that verifies a basis $B$ of $M$ with $K \subseteq B$ and $D+e \cap B = \emptyset$.
    If $e \not\in B'$, then this directly follows for $Q^*$ and the full MWB $B = K \cup B'$ of $\cM$.
    Thus, assume $e \in B'$.  We then have $e \notin \mspan[M']{B' -e}$ because $B'$ is independent.

    As $C$ is a circuit that contains $e$,
    there exists an edge $e' \in C- e$ with $e' \notin  \mspan[M']{B'-e}$.
    By assumption $e$ has maximum weight in $C$ and, thus, $w_e \ge w_{e'}$. Since $B'-e+e'$ is a basis by choice of $e' \notin  \mspan[M']{B'-e}$, we must have $w_e = w_{e'}$, as otherwise we would arrive at a contradiction to $B'$ being an MWB. 
    
    We distinguish between the following cases:
    \begin{enumerate}
        \item If $e, e' \in Q^*$, then $Q^*$ also verifies that $B = K \cup (B' -e +e')$ is an MWB for $M$ by~\Cref{lem:trivial:exchange}.

        \item If $e \in Q^*$ and $e' \not\in Q^*$, then we must have $w_e \le L_{e'}$. Otherwise, $Q^*$ would not verify that $e$ has minimum-weight in $E(M') \setminus \mspan[M']{B' - e}$. Since $L_{e'} \le w_{e'} = w_e = L_e$, we get $L_e = L_{e'}$. 
        By~\Cref{lem:extreme:lowerlimit:exchange}, this implies that  $Q = Q^*  - e +e'$ verifies the MWB $B = K \cup (B' - e + e')$.
        Note that $e'$ can be trivial, in which case we can just use $Q = Q^*  - e$ instead.

        If $e'$ is trivial, then we clearly have $c(Q) \le c(Q^*)$. Otherwise, i.e., if $e'$ is non-trivial, we have $e' \in E_{w_e}^L$ as we already argued that  that $L_{e'} = w_{e'} = w_e$. By assumption (iv) of the lemma, this implies $c_e \ge c_{e'}$ and, thus, $c(Q) \le c(Q^*)$.

        \item If $e \not\in Q^*$, then we must have $w_{e'} \ge U_e$. Otherwise $Q^*$ cannot verify that $e$ has minimum weight among elements not in $\mspan[M']{B' - e}$. However, this implies $w_{e'} = L_e \ge U_e$, which can only be the case if $e$ is trivial; a contradiction to the requirements of the lemma. \qedhere
    \end{enumerate}
\end{proof}

\lemCutTrivial*

\begin{proof}
    Let $Q^*$ denote a minimum-cost certificate that is compatible with $D$ and $K$. Then, this certificate must also verify some MWB $B'$ for $M[D,K]$.

     To prove the lemma, we have to show that there exists a minimum-cost certificate that verifies a basis $B''$ of $M$ with $K \subseteq B''$ and $D+e \cap B'' = \emptyset$.
    If $e \in B'$, then the lemma directly follows for $Q^*$ and the full MWB $B'' = K \cup B'$. Thus, assume $e \not\in B'$ and let $C$ denote the fundamental circuit  of $e$ with respect to $B'$.

    Since $C$ is a circuit that contains $e$, it must contain at least one edge $e' \in (E(M') \setminus \mspan[M']{B-e}) -e $ by~\Cref{obs:matroid:cuts}. Also, $e' \in B'-e$ because $C-e \subseteq B'$. By assumption that $e$ has minimum weight among elements not in $\mspan[M']{B-e}$ and that $B'$ is an MWB, we have $w_e = w_{e'}$.
    We distinguish the following cases:
    \begin{enumerate}
        \item If $e'\in Q^*$, then $Q^*$ also verifies the MWB $B'' = K \cup (B' - e' +e)$ by~\Cref{lem:trivial:exchange}.
        \item If $e' \not\in Q^*$, then we must have $U_{e'}\le w_e$. Otherwise, $Q^*$ would not verify that $e$ has maximum weight in $C$. However, this implies $w_{e'} = U_{e'}$. By \Cref{asm:matroids:no:non-trivial} and as $e'$ is minimal among elements in $E(M') \setminus \mspan[M']{B-e}$, this means that $e'$ must be trivial. Since both $e$ and $e'$ are trivial with  $w_e = w_{e'}$, $Q^*$ also verifies the MWB $B'' = K \cup (B' - e' + e)$ by~\Cref{lem:trivial:exchange}.
    \end{enumerate}

    In both cases, there is a certificate $Q$ with $c(Q) = c(Q^*)$ that is compatible with $D$ and $K$ and verifies an MWB $B$ of $M[D,K]$ with $e \in B$.
\end{proof}

\lemCircuitTrivial*

\begin{proof}
    Let $Q^*$ denote a minimum certificate that is compatible with $D$ and $K$. Then, this certificate must also verify some MWB $B'$ for $M[D,K]=:M'$.

     To prove the lemma, we have to show that there exists a minimum-cost certificate that verifies a basis $B$ of $M$ with $K \subseteq B$ and $D+e \cap B = \emptyset$.
    If $e \not\in B'$, then this directly follows for $Q^*$ and the full MWB $B = K \cup B'$.
    Thus, assume $e \in B'$.

    Since $C$ is a circuit that contains $e$, it must also contain some $e'\neq e$ in $B'-e$.
    By assumption that $e$ has maximum weight in $C$ and that $B'$ is an MWB, we have $w_e = w_{e'}$.
    We distinguish the following cases:
    \begin{enumerate}
        \item If $e'\in Q^*$, then $Q^*$ also verifies the MWB $B = K \cup (B' -e+e')$ by~\Cref{lem:trivial:exchange}.
        \item If $e' \not\in Q^*$, then we must have $L_{e'}\ge w_{e}$. Otherwise, $Q^*$ would not verify that $e$ has minimum weight in $E(M')\setminus \mspan[M']{B' -e}$. However, this implies $w_{e'} = L_{e'}$. By \Cref{asm:matroids:no:non-trivial} and as $e'$ is maximal on $C$, this means that $e'$ must be trivial. Since both $e$ and $e'$ are trivial with  $w_e = w_{e'}$, $Q^*$ also verifies the MWB $B = K \cup (B' -e + e')$ by~\Cref{lem:trivial:exchange}.
    \end{enumerate}
    In  both cases, there is a certificate $Q$ with $c(Q) = c(Q^*)$ that is compatible with $D$ and $K$ and verifies an MWB $B$ with $e \not\in B$.
\end{proof}

\coroExtremeCaseInstances*

\begin{proof}
    If $A_e = \{L_e,U_e\}$ for all $e \in E(M)$, then all elements have either $w_e = L_e$ or $w_e = U_e$. By~\Cref{asm:matroids:no-extrem-elements}, the compatible minor $M' := M[D,K]$ does not contain such elements. Hence, $D \cup K = E(M)$ and $E(M')=\emptyset$. Since $M'$ is a compatible minor, this implies that $K$ is an MWB that can be verified with a minimum-cost certificate.

    For the second part of the corollary, assume that $A_e=\{L,U\}$ for all $e\in E(M)$ for some $L,U$ with $L < U$. Fix an MWB $B$.
    Starting with $K,D = \emptyset$, we argue that we can iteratively add all elements of $B$ to $K$ and add all elements of $E(M)\setminus B$ to $D$ while maintaining the invariant that $M[D,K]$ is a compatible minor. In the end, we have $K = B$ which implies that $K$ can be verified with a certificate of minimum cost.

    As long as there exists an element $e \in B\setminus K$ with $w_e = U$, we argue that we can add $e$ to $K$ while maintaining our invariant. To see this, observe that $e \in B$ and $K \subseteq B$ imply that there exists an MWB $B'$ of $M[D,K]$ with $e \in B'$. Hence, $e$ must have minimum weight in $E(M') \setminus \mspan[M']{B'-e}$. Using the assumption that $e$ is non-trivial and that the query costs are uniform, this allows us to apply~\Cref{lem:matroid:cut:non-trivial} to add $e$ to $K$ while maintaining the invariant.

    As long as there exists an element $e \in E(M)\setminus (B\cup D)$ with $w_e = L$, we argue that we can add $e$ to $D$ while maintaining our invariant. To see this, observe that $e \not\in B$ and $K \subseteq B$ imply that there exists an MWB $B'$ of $M[D,K]$ with $e \not\in B'$. Hence, $e$ must have maximum weight in the fundamental circuit $C$ of $e$ with respect to $B'$. Using the assumption that $e$ is non-trivial and that the query costs are uniform, this allows us to apply~\Cref{lem:matroid:circuit:non-trivial} to add $e$ to $D$ while maintaining the invariant.

    Afterwards, all remaining elements $e \in B\setminus K$ have $w_e = L$ and all remaining elements $e \in E(M) \setminus (B \cup D)$ have $w_e = U$. Hence, $B\setminus K$ is the unique MWB of $M[D,K]$, which allows us to contract all elements of $B\setminus K$ and delete all elements of $E(M) \setminus (B \cup D)$ while maintaining our invariant.
\end{proof}

\subsection{Proof of~\Cref{lem:unique:MWB}}

In this section, we prove~\Cref{lem:unique:MWB}, which we restate here for convenience. 

\lemUniqueMWB*

To show the lemma, we characterize a set of elements $F$ that must be part of every certificate for a compatible minor $M[D,K]$ that satisfies~\Cref{asm:matroids:no-extrem-elements}. 
Since $M[D,K]$ is a compatible minor, there exists a minimum-cost certificate $Q^*$ for $M$ that is compatible with $D$ and $K$ and, thus, is a certificate for $M[D,K]$. Hence, such a certificate $Q^*$ must contain $F$. Define $\mathcal{M}'$ to be a copy of $\mathcal{M}$ that replaces the uncertainty areas $A_e$ for $e \in F$ with $A'_e = [w_e]$. As $F_e\subseteq Q^*$, the certificate $Q^*$ verifies $B$ for $\mathcal{M}$ if and only if $Q^*\setminus F$ verifies $B$ for $\mathcal{M}'$. For our goal to add further elements to $D$ and $K$ while maintaining the compatible minor invariant, this allows us to assume without loss of generality that the elements of $F$ are trivial. 
We continue by characterizing elements that belong to the set $F$. To do so, we will use the concept of a \emph{lower limit basis}, which was introduced in~\cite{MegowMS17}. 
\begin{definition}[Lower limit basis~\cite{MegowMS17}]
    Given an instance that satisfies~\Cref{asm:matroids:no-extrem-elements}, define for each $e \in E$, $w_e^L = L_e + \epsilon$ for an infinitesimally small $\epsilon > 0$. The \emph{lower limit basis} is an MWB with respect to the weights~$w^L$. 
\end{definition}

The next lemma allows us to identify an element that has to be part of every certificate for $M'$. A version of this lemma for graphic matroids with a similar proof can be found in~\cite[Lemma 4]{ErlebachLMS22}. For the sake of completeness, we give here a proof for general matroids. 

\begin{restatable}{lemma}{lemMatroidMandatory}
        \label{lem:matroid:mandatory}
    Let $M':=M[D,K]$ be a compatible minor of the uncertainty matroid $\cM$ such that~\Cref{asm:matroids:no-extrem-elements} is satisfied.
    Consider any lower limit basis $B_L$ of $M'$ and let $\{f_1,\ldots, f_\ell\}$ denote the elements of $E(M') \setminus B_L$  ordered by non-decreasing lower limit, i.e., $j < i$ implies $L_j \le L_i$.
    Assume that there exists an index $j \in \{1,\ldots, \ell\}$ such that $w_{f_j} < U_e$ for some $e \in C - f_j$, where $C$ is the fundamental circuit 
    of $f_j$ with respect to $B_L$, and let $j$ be the smallest such index.
    Then, the element $e$ with maximum upper limit in $C - f_j$ is part of every certificate for $M[D,K]$.    
\end{restatable}

\begin{proof}
First, consider any $f_i$ with $i \in \{1,\ldots,j-1\}$. By choice of $j$, we have $w_{f_i} \ge U_e$ for all $e \in C-f_i$ for the fundamental circuit of $f_i$ with respect to $B_L$. 
As $M'$ satisfies~\Cref{asm:matroids:no-extrem-elements}, this implies $w_{f_i} > w_e$ for all $e \in C -f_i$.
By the proof of~\Cref{obs:matroid:unique:max}, this implies that $f_i$ is not part of any MWB of $M'$.
Since this holds for every $f_i$ with $i \in \{1,\ldots,j-1\}$, we can delete all elements of $\{f_1,\ldots,f_{j-1}\}$ without loss of generality. Each certificate for $M'$ must still be  a certificate after these deletions.

Finally, consider the element $f_j$ and the fundamental circuit $C$ of $f_j$ with respect to $B_L$ .  Let $e$ be the element of $C- f_j$ with the maximum upper limit. By assumption of the lemma, $w_{f_j} < U_e$. 
Furthermore, we have  $L_e \le L_{f_j}$ since $B_L$ is a lower limit basis.

For the sake of contradiction, assume that there is a certificate for $M'$ that does not contain $e$.  Then $Q = E(M') - e$ must also be a certificate for $M'$. 

Consider the weight assignment $w^1$ with $w^1_e = w_e^L = L_e + \epsilon$ and $w^1_{e'}= w_{e'}$ for all $e' \not= e$. In this weight assignment, we have $w^1_e \le w^L_{f_j} = L_{f_j} + \epsilon < w_{f_j} = w^1_{f_j}$, where the strict inequality uses that the instance satisfies~\Cref{asm:matroids:no-extrem-elements}. Using that $f_j$ has minimum lower limit in $E(M')\setminus B_L$ after we deleted  $\{f_1,\ldots,f_{j-1}\}$, this implies that $e$ has unique minimum weight in $E(M') \setminus \mspan[M']{B_L  - e}$ with respect to the weight assignment $w^1$. Thus, $e$ is part of every MWB for the weights~$w^1$.

Next, consider the weight assignment $w^2$ with $w^2_e = U_e - \epsilon$ and $w^2_{e'}= w_{e'}$ for all $e' \not= e$. In this weight assignment, we have $w^2_e > w^2_{f_j}$ by our assumption that $w_{f_j}< U_e$. 
Using that $e$ has maximum upper limit in  $C - f_j$, we also get $w^2_e > w^2_{e'}$ for every $e' \in C \setminus \{e,f_j\}$. This implies that $e$ has unique maximum weight in $C$ with respect to the weights $w^2$ and is not part of any MWB for the weights $w^2$.

By~\Cref{def: what it means for S to verify T}, this means that $Q$ does not verify a basis $B$ with $e \in B$ and it also does not verify a basis $B$ with $e \not\in B$. We can conclude that $Q$ does not verify any basis and, thus, is not a certificate for $M'$. Hence, every certificate for $M'$ must contain $e$.
\end{proof}

\Cref{lem:matroid:mandatory} implies the following corollary.

\begin{coro}
\label{coro:mandatory}
    Let $M'=M[D,K]$ be a compatible minor that satisfies~\Cref{asm:matroids:no-extrem-elements}.
    Let $e$ be an element that satisfies~\Cref{lem:matroid:mandatory}. Then, there exists a minimum-cost certificate $Q$ with $e \in Q$ that is compatible with $D$ and $K$.
\end{coro}

Next, we observe that instances $M[D,K]=:M'$ that satisfy~\Cref{asm:matroids:no-extrem-elements} and do not have any elements that satisfy~\Cref{lem:matroid:mandatory} have a unique MWB.
The proof of the following lemma is inspired by 
\cite[Lem.~5]{ErlebachLMS22} for graphic matroids.
 
\begin{lemma}
    \label{lem:unique}
    Let $M':= M[D,K]$ be a compatible minor that satisfies~\Cref{asm:matroids:no-extrem-elements} such that no element in $E(M')$ satisfies~\Cref{lem:matroid:mandatory}. Then, $M'$ has a unique MWB.
\end{lemma}

\begin{proof}
Let $B_L$ denote a lower limit basis of $M'$. We claim that $B_L$ is the unique MWB for $M'$. 
To prove this claim, we show that each $f \in E(M') \setminus B_L$ has unique maximum weight in the fundamental circuit $C$ of~$f$ with respect to $B_L$. By the proof of~\Cref{obs:matroid:unique:max}, this implies that $B_L$ is a unique MWB of~$M'$. 

Consider an arbitrary $f \in E(M') \setminus B_L$ and consider the fundamental circuit $C$. By assumption, no element satisfies~\Cref{lem:matroid:mandatory} which implies $w_f \ge U_e$ for every $e \in C - f$.
Since $M'$ has, by assumption, no extreme case elements, this immediately implies $w_f > w_e$ for each  $e \in C - f$. Thus, $f$ has unique maximum weight in $C$. 
\end{proof}

With this~\Cref{coro:mandatory} and~\Cref{lem:unique} in place, we are ready to prove~\Cref{lem:unique:MWB}. 

\begin{proof}[Proof of~\Cref{lem:unique:MWB}]
    Fix an arbitrary MWB $B$ of $M' = M[D,K]$. To show the lemma, we argue that each element of $B$ (element of $E(M')\setminus B$) can be contracted (deleted) and added to $K$ (to $D$) while maintaining the invariant that $M[D,K]$ is a compatible minor of $\cM$. 

    As long as there exists an element $e$ that satisfies~\Cref{lem:matroid:mandatory}, we can use~\Cref{coro:mandatory} and without loss of generality replace the uncertainty area of $e$ with $\{w_e\}$. If $e \in B$, then $e$ must have minimum weight in $E \setminus \mspan[M']{B-e}$ and we can apply~\Cref{lem:matroid:cut:trivial} to contract $e$ and add it to $K$ while maintaining the invariant that $M' = M[D,K]$ is a compatible minor of $\cM$. If $e \not\in B$, then $e$ must have maximum weight in the fundamental cycle $C_e$ of $e$ with respect to $B$. Thus, we can apply~\Cref{lem:matroid:circuit:trivial} to delete $e$ and add it to $D$ while maintaining the invariant that $M' = M[D,K]$ is a compatible minor of $\cM$. In both cases, the new compatible minor still satisfies~\Cref{asm:matroids:no-extrem-elements}. Hence, we can repeat this process until no remaining element satisfies~\Cref{lem:matroid:mandatory}.

    Let $M' = M[D,K]$ be the resulting compatible minor. Note that $D \cap B = \emptyset$, which implies that $B' = B\setminus K$ is an MWB of $M'$. By~\Cref{lem:unique}, $M'$ has a unique MWB and, thus, $B'$ must be the unique MWB of $M'$. This allows us to use~\Cref{obs:matroid:unique:min} to contract all elements of $B'$ and add them to $K$ while maintaining our invariant. Similarly, we can use~\Cref{obs:matroid:unique:max} to delete all remaining elements and add them to $D$ while maintaining our invariant. 
\end{proof}

\subsection{Running time of~\Cref{alg:basis}}

\begin{restatable}{lemma}{lemBasisRunningTime}
    \label{lem:alg:basis:running-time}
    \Cref{alg:basis} runs in polynomial time.
\end{restatable}

\begin{proof}
The values of $L_e$ and $U_e$ can be easily computed at the beginning for each $e \in E$, using the list of intervals for each $A_e$.  Subsequently,
    no further information is needed about $A_e$.
    
    Since the number of iterations of the while-loop is bounded by the number of elements, it only remains to argue that each iteration can be executed in polynomial time. To this end, we mainly have to argue that, for an element $e$, we can check in polynomial time whether the prerequisites of ~\Cref{obs:matroid:unique:max} and \ref{obs:matroid:unique:min} or Lemmas~\ref{lem:matroid:circuit:non-trivial} to 
    \ref{lem:matroid:mandatory} are satisfied.

    We separately argue about the different lemmas and observations. Assume there is given a compatible minor $M':= M[D,K]$ of an uncertainty matroid $\cM$.
    \begin{enumerate}[(a)]
        \item \Cref{obs:matroid:unique:min}: {Checking whether an element $e$ satisfies the condition of~\Cref{obs:matroid:unique:min} is equivalent to checking whether $e$ is part of every MWB of $M'$.} 
    To check whether an element $e$ is part of every MWB
        w.r.t.~the weights~$w$, we can just compute some minimum-weight basis $B$ using the standard greedy algorithm. If $e \not\in B$, then~\Cref{obs:matroid:unique:min} clearly does not apply. Otherwise, i.e., $e \in B$, we can check whether $e$ can be exchanged with an element $e' \in E \setminus B$ with $w_{e'} \le w_e$. If this is possible, then $e$ is not part of every minimum-weight basis. If not, then $e$ is part of every minimum-weight basis. The running time of this test is dominated by the running time of the greedy algorithm and, thus, polynomial in the input size.
        \item \Cref{obs:matroid:unique:max}: Using matroid duality arguments allows us to argue in the same way as in the previous case.
        \item \Cref{lem:matroid:cut:non-trivial}: Checking whether there is an element $e$ that is \emph{non-trivial} and (i) there exists a basis $B$ of $M'$ with $e \in B$, (ii) $w_e = U_e$, (iii) $e$ has minimum weight in $E(M')\setminus \mspan[M']{B-e}$  and (iv) $e$ has maximum query cost in $E(M')\setminus \mspan[M']{B-e} \cap E_{w_e}^U$, can also be done in polynomial time. Note that, a non-trivial element $e$ that satisfies (i)-(iii) is a non-trivial element $e$ with $w_e=U_e$ that is in an MWB $B$. Given such $e$ we can find the MWBes that contain $e$ by computing some MWB and checking whether $e$ is included or trying to add it by exchanging it with an equal-weight element. Given such a pair of $e$ and MWB $B \ni e$, we can easily check whether $e$ has maximum query cost in $(E(M')\setminus \mspan[M']{B-e}) \cap E_{w_e}^U$. If this is the case, then we are done. If not, then some other element of $(E(M')\setminus \mspan[M']{B-e}) \cap E_{w_e}^U$ satisfies the condition of~\Cref{lem:matroid:cut:non-trivial} and we are done anyway.
        
        \item \Cref{lem:matroid:circuit:non-trivial}: Using matroid duality arguments allows us to argue in the same way as in the previous case.
                
        \item The conditions of \Cref{lem:matroid:cut:trivial} and \ref{lem:matroid:circuit:trivial} can be checked similarly to \cref{obs:matroid:unique:max} and \ref{obs:matroid:unique:min} plus checking \Cref{asm:matroids:no:non-trivial}, which is trivial.
        
        \item The final step of the algorithm consists of computing an MWB of the minor $M[D,K]$ which can be done by the well-known Greedy algorithm in polynomial time.
    \qedhere
    \end{enumerate}
\end{proof}

\section{Missing proofs from \Cref{sec: construct cert}}
\label{app:construct}

\lemmaInCert*

\begin{proof}
    Let $Q$ be a certificate that verifies $B$. Consider a non-basis element $e$, 
    and the fundamental circuit $C_e$ with respect to  $B$. Let $w^*$ be some weight assignment consistent with $Q$ (Definition \ref{def: weight assignment consistent with certificate}).
    Given $w^*$, if any element $f \in C_e - \{e\}$ satisfies $U_f \leq L_e$, it immediately follows that $w^*_e \geq w^*_f$, without including either $e$ or $f$ in $Q$. Define $F_e$ as the set of elements $f \in C_e - \{e\}$ where $U_f > L_e$. Additionally, define $\hat{F_e}$ as the set of elements $f \in C_e - \{e\}$ with $U_f > w_e$. Clearly, $\hat{F}_e \subseteq F$.

    Recall that if $e \in Q$, then $w_e = w^*_e$, and if $e \notin Q$, $w^*_e \in A_e$. We aim to establish the necessary and sufficient conditions for $Q$ to satisfy Corollary~\ref{coro:certificate:characterization:circuits} for all elements $f \in F$, (i.e., for $Q$ to satisfy $L_e(Q) \ge U_f(Q)$ for all $f \in F$). Observe the following cases:

    \noindent \textbf{Case 1: } The first case is if $w_e \geq U_f$ for all $f \in F_e$. Two sub-cases arise in this scenario.

     \noindent \textbf{Subcase 1.1: } There exists some $f' \in F$ where $w_{f'} > L_e$. In this case, we claim it is necessary and sufficient to have $e \in Q$. Suppose $e \notin Q$. Then by Definition \ref{def: weight assignment consistent with certificate}, we have $w^*_e \in A_e$. Since $w_{f'} > L_e$, Corollary \ref{coro:certificate:characterization:circuits} is not satisfied. Thus, $e$ must be in $Q$. Conversely, if $e \in Q$, we have that for all $f \in F_e$, $w^*_e = w_e = L_e(Q) \geq U_f \geq U_f(Q) \geq w_f$ and thus $w^*_e \geq w_f$ and Corollary \ref{coro:certificate:characterization:circuits} is satisfied.

     \noindent \textbf{Subcase 1.2: }  For all $f \in F_e$, $w_f \leq L_e$. We show that in this case, it is necessary and sufficient to have $e \in Q$ or $F_e\subseteq Q$. Suppose $e \not\in Q$ and $F_e\not\subseteq Q$. Then there exists some element $f' \in F_e\backslash Q$. By~\Cref{def: weight assignment consistent with certificate}, we have $w^*_e \in A_e$ and $w^*_{f'} \in [L_{f'}, U_{f'}]$. Since $U_{f'} > L_e$, Corollary \ref{coro:certificate:characterization:circuits} is not satisfied. Thus, it's clear that $e \in Q$ or $F_e\subseteq Q$. Conversely, suppose $e \in Q$ or $F_e\subseteq Q$.  If $e \in Q$ we have that $w^*_e = w_e = L_e(Q) \geq U_f \geq U_f(Q) \geq w_f$ and so $w^*_e \geq w_f$.  If $F_e\subseteq Q$, we have that for all $f \in F_e$, $w_e \geq L_e(Q) \geq L_e \geq U_f(Q) = w_f = w^*_f$ and thus $w_e \geq w^*_f$. Either way, Corollary \ref{coro:certificate:characterization:circuits} is satisfied.

     \noindent \textbf{Case 2: } The second case is if $w_e < U_f$ for some $f \in F_e$. Again two sub-cases arise in this scenario.
     
     \noindent \textbf{Subcase 2.1: } There exists some $f' \in C_e - e$ where $w_{f'} > L_e$. In this case, we show that it is necessary and sufficient to have $\hat{F_e} + e \subseteq Q$. Suppose $\hat{F_e} + e \not\subseteq Q$. Then either $e \notin Q$ or there exists some $f'' \in \hat{F_e} \backslash Q$. Suppose $e \notin Q$. By Definition \ref{def: weight assignment consistent with certificate}, we have $w^*_e \in A_e$. Since there exists some $w_{f'} > L_e$, Corollary \ref{coro:certificate:characterization:circuits} is not satisfied. Suppose $f'' \notin Q$. By Definition \ref{def: weight assignment consistent with certificate}, we have $w^*_{f''} \in [L_{f''}, U_{f''}]$. Since $w_e < U_{f''}$, Corollary \ref{coro:certificate:characterization:circuits} is not satisfied. Thus, $\hat{F_e}+e \subseteq Q$. 
     
     Conversely, if $\hat{F_e} + e \subseteq Q$ we have for all $f \in \hat{F_e}$ that $w^*_e = w_e = L_e(Q) \geq U_f \geq U_f(Q) = w_f = w^*_f$. For all $f \in F_e\setminus \hat{F_e}$, given that $L_e < U_f \leq w_e$, the inclusion of $e \in Q$ ensures that $w^*_e = w_e = L_e(Q) \geq U_f \geq U_f(Q) \geq w_f$. Therefore, Corollary \ref{coro:certificate:characterization:circuits} is satisfied.
     
     \noindent \textbf{Subcase 2.2: } For all $f' \in C_e - \{e\}$, $w_{f'} \leq L_e$. In this case, we show it is necessary and sufficient to have $F_e\subseteq Q$ or $\hat{F_e} + e\subseteq Q$. Suppose neither $F_e$ nor $\hat{F_e} + e$ is contained in $Q$. Then there exists some $f'' \in \hat{F_e} \backslash Q$. By Definition \ref{def: weight assignment consistent with certificate}, we have that $w^*_{f''} \in [L_{f''}, U_{f''}]$. Since $w_e < U_{f''}$, Corollary \ref{coro:certificate:characterization:circuits} is not satisfied. Thus, it's clear that all such $f''$ must be in $Q$. Now suppose some $f'' \in F_e\backslash \hat{F_e}$ is not in $Q$. This element has the property that $w_e \geq U_{f''} > L_e \geq w_{f''}$.  By Definition \ref{def: weight assignment consistent with certificate} we have that $w^*_{f''} \in [L_{f''}, U_{f''}]$. Since $L_e < U_{f''}$, Corollary \ref{coro:certificate:characterization:circuits} is not satisfied. However, having $e$ or all such $f'' \in Q$ will show that $w^*_e > w^*_{f''}$. 
    
    By ensuring that $F_e\subseteq Q$ or $\hat{F_e} + e\subseteq Q$ we have that either $w_e \geq L_e(Q) \geq L_e \geq U_f(Q) = w_{f} = w^*_{f}$ and so $w_e \geq w^*_{f}$ for all $f \in F$ or $w^*_e = w_e = L_e(Q) \geq U_{f} \geq U_f(Q) = w_{f} = w^*_{f}$ so $w^*_e \geq w^*_{f}$ for all $f'' \in \hat{F_e}$. For all $f'' \in F_e\setminus \hat{F_e}$, given that $L_e < U_{f} \leq w_e$, the inclusion of $e \in Q$ ensures that $w^*_e = w_e = L_e(Q) \geq U_{f} \geq U_f(Q) \geq w_{f}$. Therefore, Corollary \ref{coro:certificate:characterization:circuits}~is~satisfied.
\end{proof}

\coroVertexCover*

\begin{proof}
    Let $Q\subseteq E$ and fix an arbitrary $e \not\in B$. We argue that the conditions of~\Cref{lemma: cert contains for Le Ue} are satisfied for $C_e$ if and only if $Q$ covers the edges in $E^B_e$. Since $E^B = \bigcup_{e\in E \setminus B} E^B_e$, this then implies that $Q$ satisfies~\Cref{lemma: cert contains for Le Ue} if and only if $Q$ is a vertex cover. We separately consider the different cases for $C_e$:
    \begin{enumerate}
        \item If $w_e \geq U_f$ for all $f \in C_e-e$, and there is a $f' \in C_e - e$ such that $w_{f'} > L_e$,  then $\{e,e\} \in E^B$ by~\Cref{def:auxiliary_graph}. Thus, $e \in Q$ holds if and only if $Q$ covers $E_e^B$.
        \item If $w_e \geq U_f$  for all $f \in C_e-e$, and $w_f \leq L_e$ for all $f \in C_e-e$, then $E^B_e = \{ \{e,f\} \mid f \in F_e\}$ by~\Cref{def:auxiliary_graph}. Hence $e \in Q$ or $F_e \subseteq Q$ holds if and only if $Q$ covers $E_e^B$.
        \item If there exists $f \in C_e - e$ such that $w_e < U_f$, and there exists $f' \in C_e-e$ (could have $f' = f$) such that $w_{f'} > L_e$, then $E^B_e=\{\{f,f\} \mid f \in \hat{F_e} + e\}$ by~\Cref{def:auxiliary_graph}. Hence $\hat{F}_e + e \subseteq Q$ if and only if $Q$ covers~$E^B_e$.
        \item If there exists $f \in C_e-e$ such that $w_e < U_{f}$, and for all $f' \in C_e -e$ it holds that $w_{f'} \leq L_e$, then $E_e^B = \{\{e,f\} \mid f \in F_e \setminus \hat{F}_e\} \cup \{\{f,f\} \mid f \in \hat{F}_e\}$. The corresponding condition of~\Cref{lemma: cert contains for Le Ue} is $F_e\subseteq Q$ or $\hat{F_e} + e\subseteq Q$. Since  $\hat{F_e} \subseteq F_e$, the condition holds if and only if $\hat{F}_e \subseteq Q$ and $F_e\setminus \hat{F}_e \in Q$ or $e \in Q$. The latter holds if and only if $Q$ covers~$E^B_e$. \qedhere
    \end{enumerate}
\end{proof}

\thmFinalAlg*

\begin{proof}
    The correctness of the algorithm immediately follows 
    \Cref{coro:vertex:cover}. Thus, it only remains to argue about the running-time. The MWB $B$ can be computed in polynomial time by~\Cref{lem:alg:basis:running-time}. Based on $B$, the graph $G^B$ can also be computed in polynomial time. 

    To show that the minimum-weight vertex cover can be computed in polynomial-time, note that $G^B$ is bipartite except for the self-loop edges. Indeed, all non-loop edges have one endpoint in $B$ and one endpoint in $E\setminus B$.  To still compute a minimum-weight vertex cover for $G^B$ using the algorithms for bipartite graphs~(see, e.g.,~\cite{Schrijver2003book}), we define $\bar{G}^B$ as a copy of $G^B$ that replaces each loop-edge $\{e,e\} \in E^B$ with an edge $\{e,v_e\}$ to a distinct new vertex $v_e$ with weight $w(v_e) = w_e + \epsilon$ for some $\epsilon > 0$. Clearly, the size of $\bar{G}^B$ is polynomial in the size of $G^B$ and a set $Q$ is a minimum-weight vertex cover of $G^B$ if and only if $Q$ is a minimum-weight vertex cover for $\hat{G}^B$. Since $\hat{G}^B$ is bipartite, we can compute a minimum-weight vertex cover for $\hat{G}^B$, and thus for $G^B$, by using the classical algorithm for bipartite graphs.
\end{proof}

\section{Missing proofs of~\Cref{sec: applications}}
\label{app:applications}

\obsNaive*

\begin{proof}
    Let $Q$ denote the query set computed by the algorithm and let $Q^*$ denote the minimum-cardinality query set. 
    If the predictions are correct, $w_e = \hat{w}_e$ for all $e \in E$, then the second step of the algorithm computes and queries a minimum-cardinality certificate for the instance. Hence, Step~2 is never executed and the algorithm is $1$-consistent.
    If the predictions are incorrect and the algorithm queries at least one element, then also $|Q^*| \ge 1$. Since the algorithm queries at most $n$ elements, $n$-robustness follows immediately.
\end{proof}

\begin{restatable}{observation}{obsExistence}
\label{obs:existence}
    The MWB $\hat{B}_s^*$ as defined in~\Cref{sec:basis:preds} always exists.
\end{restatable}


\begin{proof}
    To argue that MWB $\hat{B}_s^*$ always exists, it suffices to show that there exists at least one MWB $B$ with $\hat{B}_s \subseteq B$.
    If $\hat{B}_s = \emptyset$, then any basis $B$ contains $\hat{B}_s$ and the statement is clearly true. Hence, assume that $\hat{B}_s$ contains at least one element.

    By definition of~$\hat{B}_s^*$, each element $f \in \hat{B}_s^*$ is a member of the basis $\hat{B}$ and appears in some correct fundamental circuit $C_e$, $e \in E \setminus \hat{B}$, w.r.t.~$\hat{B}$. By definition of correct circuits, this means that $f$ satisfies $w_f \le w_g$ for all $g \in E \setminus \mspan[M]{\hat{B} - f}$. Hence, $f$ has minimum weight in $E \setminus \mspan[M]{\hat{B} - f}$, which means that there exists at least one MWB $B$ with $f \in B$. Since this argument holds for all $f \in \hat{B}_s$ and $\hat{B}_s$ is independent (as it is a subset of the basis $\hat{B}$), this implies that there exists an MWB $B$ with $\hat{B}_s \subseteq B$: Formally, we can contract the elements of $\hat{B}_s$ one after the other while maintaining the invariant that there exists an MWB of the original matroid consisting of the elements that have been contracted so far and an MWB of the resulting minor.
\end{proof}

\subsection{Proof of~\Cref{thm:basis:predictions}}

In this section, we prove~\Cref{thm:basis:predictions}, which we restate here for convenience.

\thmBasisPredictions*

We remark that the factor $c_{\max}$ of the error $\eta_2$ is indeed necessary: There are instances where all incorrect fundamental circuits $C_e$ are disjoint, and the algorithm has to identify the maximum-weight element in each $C_e - e$ as the incorrect $e$ does not have maximum-weight in $C_e$. The lower bound example given in ~\cite[Section~7]{GuptaSS16} shows that an adversary can force an algorithm to query all elements of $C_e - e$ while the optimum only queries a single one.

To prove the theorem, we consider the following algorithm.
This algorithm is a slight variation of the one given in~\Cref{sec:fixed:trees}, taking into  account that elements $e \in E \setminus \hat{B}$ might not be maximal on the circuit $C_e$. In such cases, the algorithm deviates from the procedure in~\Cref{sec:fixed:trees} by querying all elements in the circuit $C_e$ (cf.~Step~2biii).

\begin{enumerate}
        \item Initialize $Q = \emptyset$.
        \item For each $e \in E \setminus \hat{B}$:
        \begin{enumerate}
            \item Let $C_e$ denote the fundamental circuit of $e$ with respect to $\hat{B}$.
            \item While there exists an $f \in C_e - e$ with $U_f(Q) > L_e(Q)$:
            \begin{enumerate}
                \item If there is an element $g \in C_e \setminus (\{e\} \cup Q)$ with $U_{g}(Q) > L_e(Q)$, then let $\hat{g}$ denote such an element with maximum $U_{\hat{g}}(Q)$. Add $\hat{g}$ to $Q$ and query $\hat{g}$.
                \item If $e \not\in Q$, then add $e$ to $Q$ and query $e$.
                \item If $w_e < w_f$ for some $f \in C_e \cap Q$, then query $C_e$ and add $C_e$ to $Q$. 
            \end{enumerate}
        \end{enumerate}
        \item Return $Q$.
    \end{enumerate}

We begin by showing that the query set $Q$ computed by the algorithm is indeed a certificate that verifies some MWB $B$.

\begin{lemma}
    \label{lem:basis:pred:1}
    The query set $Q$ computed by the algorithm is a certificate that verifies some MWB $B$.
\end{lemma}

\begin{proof}
    We begin by constructing an MWB $B$ and then argue that $Q$ indeed verifies $B$. To this end, we first construct a minor $M'$ by executing the following deletions and contractions:
    \begin{itemize}
        \item Delete all elements $e \in E \setminus \hat{B}$ for which the algorithm does \emph{not} execute queries in Step~2biii.
        \item Contract all elements $f \in \hat{B}$ that are only part of fundamental circuits $C_e$, $e \in E \setminus \hat{B}$, for which the algorithm does \emph{not} execute queries in Step~2biii.
    \end{itemize}
    By definition, any deleted element $e$ has maximum weight in its fundamental circuit~$C_e$. Hence, there exists an MWB $B$ that does not contain any of the deleted elements $e$. Furthermore, by definition of the algorithm, we have $L_e(Q) \ge U_f(Q)$ for all $f \in C_e$. This means querying~$Q$ reveals sufficient information to conclude that all such elements $e$ can be deleted.

    Similarly, all contracted elements $f$ satisfy that $E \setminus \mspan[M]{\hat{B} - f}$ only contains $f$ and a subset of the deleted elements. Otherwise, $f$ would be part of a circuit $C_{e'}$ for which the algorithm executes queries in Step 2biii. This also means that all such elements $f$ have minimum  weights
    in $E \setminus \mspan[M]{\hat{B} - f}$, as otherwise they would need to be part of a circuit $C_{e'}$ for which the algorithm executes queries in Step 2biii. Hence, there exists an MWB that contains all the contracted elements. Furthermore, querying $Q$ reveals sufficient information to conclude that all such elements can be contracted (this is a consequence of the fact that  $L_e(Q) \ge U_g(Q)$ holds for all $g \in C_e$ for all deleted elements $e$).

    The arguments above imply that there exists an MWB $B$ that consists of the set of contracted elements $K$ and some MWB $B'$ of the minor $M'$. By assumption, the algorithm queries all elements of the minor $M'$ and, thus, querying $Q$ clearly reveals sufficient information to verify an MWB $B'$ of $M'$.
\end{proof}

Next, we show that $Q$ satisfies the guarantee of~\Cref{thm:basis:predictions}.

\begin{lemma}
    \label{lem:basis:pred:2}
    The query set $Q$ computed by the algorithm satisfies $|Q| \le 2 \cdot (|Q^*|+\eta_1)+\eta_2 \cdot c_{\max}$.
\end{lemma}

\begin{proof}
 Let $E\setminus \hat{B} = \{e_1, \ldots, e_k\}$ be indexed in the order in which the algorithm considers the elements of $E \setminus \hat{B}$ in the loop of Step 2. For each $j \in \{1,\ldots, k\}$ let $Q_j$ denote the set of elements that the algorithm queries during iteration $j$ of the loop of Step 2.

 Let $\hat{B}'_s$ denote the subset of $\hat{B}$ such that each element of  $\hat{B}'_s$ appears on a correct fundamental circuit w.r.t.~$\hat{B}$, and let $B^*_s$ denote the MWB with minimum verification cost among all MWB's that contain the set $\hat{B}_s'$. By~\Cref{obs:existence} such an MWB $B^*_s$ always exists.
 Finally, let $Q^*_s$ denote a minimum-cardinality certificate for verifying $B_s^*$.

We show that, for each $Q_j$ with a correct fundamental circuit $C_{e_j}$, it holds 
\begin{equation}
\label{eq:error:1}
\frac{|Q_j|}{|Q_s^* \cup Q_j|} \le 2.    
\end{equation}

Since the $Q_j$'s are pairwise disjoint, each $Q_j$ with incorrect $C_{e_j}$ contributes one to $\eta_2$, and each $Q_j$ satisfies $|Q_j| \le c_{\max}$,  the inequality~\eqref{eq:error:1} implies
\begin{align*}
    |Q| \le 2 \cdot |Q^*_s| + \eta_2 \cdot c_{\max} \le 2 \cdot |Q^*| + 2 \cdot \eta_1 + \eta_2 \cdot c_{\max},
\end{align*}
where the second inequality follows from the definition of $\eta_1$. Hence, it remains to show that~\eqref{eq:error:1} holds for each $Q_j$ with a correct fundamental circuit $C_{e_j}$.

This proof will essentially be a repetition of the proof of~\Cref{lem:comp:knowntree}. We will exploit that each correct fundamental circuit $C_{e_j}$ with respect to $\hat{B}$ is also a fundamental circuit of $B^*_s$ by definition of $B^*_s$. Hence $Q^*_s$ has to satisfy the conditions of~\Cref{lemma: cert contains for Le Ue} for each such $C_{e_j}$.

 Fix an arbitrary $Q_j$ with correct $C_{e_j}$. Define $P_j = \bigcup_{j' < j} Q_{j'}$.
 We show that $Q_j$ satisfies~\eqref{eq:error:1} via case distinction over the cardinality of $Q_j$.
    \begin{enumerate}
        \item If $|Q_j| = 0$, then~\eqref{eq:error:1} holds trivially.
        \item If $|Q_j| = 1$, then the definition of the loop of Step 2b implies that we must have $Q_j = \{e_j\}$ as $e_j \in Q_j$ holds by Step 2bii of the algorithm. The fact that the loop of Step 2b is executed but only queries $e_j$ implies that all $g \in C_{e_j} \setminus (\{e\} \cup P_j)$ have $w_g \le U_g(P_j) \le L_{e_j}(P_j) = L_{e_j}$ (cf.~Step 2bi). However, then the fact that the loop of Step 2b is executed implies that there is some $g \in C_{e_h} \cap P_j$ with $U_g(P_j) = w_g \ge L_{e_j}(P_j) = L_{e_j}$. This in turn implies that $C_{e_j}$ either satisfies the first or third case of~\Cref{lemma: cert contains for Le Ue}. Hence, $e \in Q^*_s$ and $\frac{|Q_j|}{|Q_j \cap Q^*_s|}=1$.
        \item If $|Q_j| = 2$, then~\Cref{lemma: cert contains for Le Ue} implies $|Q_j \cap Q^*_s| \ge 1$. Hence, $\frac{|Q_j|}{|Q_j \cap Q^*_s|}= \le 2$
        \item If $|Q_j| > 2$, then the loop of Step 2b is executed at least twice. 
        For the two elements that are queried in the first iteration, we can argue that $Q^*_s$ has to contain at least one of them as in the previous case.
        
        Let $Q_j'$ denote the subset of $Q_j$ that is queried after the first iteration of this Step 2b loop. All $g \in Q_{j}'$ have to satisfy $U_g > w_{e_j}$, as otherwise the loop would terminate before $g$ is queried. This implies that all $g \in Q_j'$ are in $\hat{F}_{e_j}$ (defined as in~\Cref{lemma: cert contains for Le Ue}). Hence, $C_{e_j}$ either satisfies case 3 or 4 of~\Cref{lemma: cert contains for Le Ue} and, thus, $Q_j' \subseteq \hat{F}_{e_j} \subseteq Q^*$. We can conclude with $\frac{|Q_j|}{|Q_j \cap Q^*_s|}\le \frac{2+ |Q_j'|}{1+ |Q_j'|} \le 2$. \qedhere
    \end{enumerate}    
\end{proof}

\Cref{lem:basis:pred:1,lem:basis:pred:2} and the observation that the algorithm never executes more than~$n$ queries imply~\Cref{thm:basis:predictions}.

\end{document}